\documentclass{article}

\usepackage[utf8]{inputenc}
\usepackage{amsmath,amsfonts,amssymb,amsthm}
\usepackage{enumerate}
\usepackage{graphicx}
\usepackage{geometry}
\usepackage{hyperref}
\hypersetup{
    colorlinks,
    linkcolor={red!50!black},
    citecolor={blue!50!black},
    urlcolor={blue!80!black}
}

\hypersetup{pdftitle=Towards a better notation for EFGs}
\usepackage[]{algorithm2e}	
\usepackage{enumitem}		
\makeatletter				
\def\namedlabel#1#2{\begingroup
    #2%
    \def\@currentlabel{#2}%
    \phantomsection\label{#1}\endgroup
}
\makeatother

\newcommand{\N}{\mathbb{N}}
\newcommand{\R}{\mathbb{R}}

\newcommand{\mc}{\mathcal}

\newcommand{\mf}{\mathbf}
\DeclareMathAlphabet{\mathpzc}{OT1}{pzc}{m}{it}

\newcommand{\vknote}[1]{\textnormal{{\color{blue} \textsf{[VK]: {#1}}}}}
\renewcommand{\vknote}[1]{}

\newcommand{\notetoself}[1]{\textnormal{{\color{red} \textsf{[N2S]: {#1}}}}}
\renewcommand{\notetoself}[1]{}

\providecommand{\corollaryname}{Corollary}
\providecommand{\claimname}{Claim}
\providecommand{\definitionname}{Definition}
\providecommand{\lemmaname}{Lemma}
\providecommand{\notationname}{Notation}
\providecommand{\remarkname}{Remark}
\providecommand{\problemname}{Problem}
\providecommand{\propositionname}{Proposition}
\providecommand{\examplename}{Example}
\providecommand{\exercisename}{Exercise}
\providecommand{\theoremname}{Theorem}
\providecommand{\conjecturename}{Conjecture}

\newtheorem{theorem}{\protect\theoremname}[section]
\newtheorem{definition}[theorem]{\protect\definitionname}
\newtheorem{lemma}[theorem]{\protect\lemmaname}

\newtheorem{example}[theorem]{\protect\examplename}
\newtheorem{exercise}[theorem]{\protect\exercisename}
\newtheorem{corollary}[theorem]{\protect\corollaryname}
\newtheorem{remark}[theorem]{\protect\remarkname}

\newtheorem{notation}[theorem]{\protect\notationname}

\newtheorem{conjecture}[theorem]{\protect\conjecturename}

\usepackage{tikz}
\usepackage{graphicx}

\usepackage{xcolor}
\usepackage{forest} 
	\usetikzlibrary{shapes} 
	\usetikzlibrary{calc} 
	\usetikzlibrary{arrows}
\usepackage{xparse}	

\newlength{\nodesize}
\colorlet{chance_color}{black}
\colorlet{pl0_color}{chance_color}
\colorlet{chance_text}{white}
\colorlet{pl1_color}{magenta!50}
\colorlet{pl2_color}{cyan!50}
\colorlet{pl0_infoset_color}{pl0_color}
\colorlet{pl1_infoset_color}{magenta}
\colorlet{pl2_infoset_color}{cyan}
\forestset{
	basenode/.style = {draw,
		inner sep = 0,
		outer ysep = 0,
		minimum size = \nodesize,
		anchor = north
	},
	playernode/.style={basenode, 
		shape = regular polygon,
		regular polygon sides = 3,
	},
	pl1/.style={playernode, fill=pl1_color},
	pl2/.style={playernode, fill=pl2_color, shape border rotate=180},
	chance/.style = {basenode,
		fill=pl0_color, text=chance_text,
		circle,
		minimum size=0.75*\nodesize,
		l sep=0.4765\nodesize,
	},
	terminal/.style = {basenode,
		shape = regular polygon,
		regular polygon sides = 4,
		l sep=0.47\nodesize,
		minimum size = 1.07\nodesize
	}
}
\tikzset{
	partition/.style = {
		draw,
		rounded corners = 5,
		inner sep=0.1\nodesize,
	},
	infoset/.style = {
		partition,
		draw=pl1_infoset_color,
	},
	augmented/.style = {
		dashed
	},
	opponent/.style = {
		draw=pl2_infoset_color,
		inner sep=0.225\nodesize,
	},
	pl1_cl_infoset/.style = {infoset, yshift=-0.035\nodesize},
	pl2_cl_infoset/.style = {infoset,
		opponent,
		inner sep=0.1\nodesize,
		yshift=0.035\nodesize
	},
	public_state/.style = {
		partition,
		draw=pl0_infoset_color,
		inner sep=0.125\nodesize,
		rounded corners = 7,
	},
}
\tikzset{
	line_infoset/.style = {
		draw=pl1_infoset_color,
		dashed
	},
	opp_line_infoset/.style = {
		draw=pl2_infoset_color,
		dashed
	}
}
\newcommand{\corners}[1]{#1.corner 1)(#1.corner 2)(#1.corner 3}

\pgfdeclarelayer{bg}    
\pgfsetlayers{bg,main}  
\usepackage{graphicx}
\usepackage{caption}
\usepackage{subcaption}
\usepackage{tikz} 
	\usetikzlibrary{shapes} 
	\usetikzlibrary{calc} 
\usepackage{xparse}	
	\usetikzlibrary{shapes.multipart} 
	\usetikzlibrary{decorations.pathmorphing} 
	\usetikzlibrary{decorations.pathreplacing} 
\colorlet{chance_color}{black}
\colorlet{pl0_color}{chance_color}
\colorlet{chance_text}{white}
\tikzset{
	basenode/.style = {draw,
		inner sep = 0,
		minimum size = \nodesize
	},
	playernode/.style={basenode, 
		shape = regular polygon,
		regular polygon sides = 3
	},
	pl1/.style={playernode, fill=pl1_color},
	pl2/.style={playernode, fill=pl2_color, shape border rotate=180},
	chance/.style = {basenode,
		fill=pl0_color, text=chance_text,
		circle,
		minimum size=0.7*\nodesize,
	},
	terminal/.style = {basenode,
		draw=none,
		outer sep=0,
		minimum size = 0.6\nodesize
	}
}
\NewDocumentCommand{\IS}{ m m O{0.6} O{0.6} O{0.6} O{5} O{dashed} O{} O{1} }{
	\draw [#7, rounded corners=#6, draw=pl#9_color!75!black]
		($(#1)+(-#5\nodesize,#3\nodesize)$) rectangle
		($(#2)+(#5\nodesize,-#4\nodesize)$);
	\node at ($(#1)!.5!(#2)$) {#8};
}
\NewDocumentCommand{\oneNodeIS}{ m O{0.6} O{0.6} O{0.6} O{5} O{dashed} O{} O{1} }{
	\IS{#1}{#1}[#2][#3][#4][#5][#6][#7][#8]
}
\NewDocumentCommand{\actIS}{ m m O{0.6} O{0.6} O{0.6} O{5} O{} O{1} }{
	\IS{#1}{#2}[#3][#4][#5][#6][solid][#7][#8]
}
\NewDocumentCommand{\oppActOneNodeIS}{ m O{0.6} O{0.6} O{0.6} O{5} O{solid} O{} }{
	\IS{#1}{#1}[#2][#3][#4][#5][#6][#7][2]
}
\NewDocumentCommand{\oppIS}{ m m O{0.6} O{0.6} O{0.6} O{5} O{dashed} O{} }{
	\IS{#1}{#2}[#3][#4][#5][#6][#7][#8][2]
}
\NewDocumentCommand{\publicState}{ m m O{0.6} O{0.6} O{0.6} O{0} O{}}{
	\IS{#1}{#2}[#3][#4][#5][#6][solid][#7][0]
}
\tikzset{
	level 1/.style = {level distance=2\nodesize, sibling distance=3\nodesize},
	level 2/.style = {level distance=2\nodesize, sibling distance=1.5\nodesize}
}
\def\infosetDeprecated[#1,#2,#3](#4,#5)(#6,#7,#8)(#9){	
	\draw [#1, rounded corners=#2, draw=pl#3_color!75!black]
		($(#4)+(-#6\nodesize,#7\nodesize)$) rectangle
		($(#5)+(#6\nodesize,-#8\nodesize)$);
	\node at ($(#4)!.5!(#5)$) {#9};
}
\def\actInfosetDeprecated(#1,#2){
	\infosetDeprecated[solid,5,1](#1,#2)(0.6,0.6,0.5)()
}
\def\augInfosetDeprecated(#1,#2){
	\infosetDeprecated[dashed,5,1](#1,#2)(0.6,0.5,0.6)()
}

\def\sneakingGamePartition{
\setlength{\nodesize}{2.5em}
\tikzset{
	level 1/.style = {level distance=2\nodesize, sibling distance=3\nodesize},
	level 2/.style = {level distance = 1.5\nodesize, sibling distance=1.5\nodesize},
	level 3/.style = {level distance = 1.5\nodesize},
	level 4/.style = {level distance = 1.75\nodesize}
}
\node(root) at ($(middle)+(-\linewidth/3+1.35em,0)$) [chance,label=right:{}]{}
	child{node(H1)[pl1]{}
		child{node{}
		edge from parent[draw=none]{}
			child{node(H2)[pl1]{}
				edge from parent[draw=none]{}
				child{node[terminal]{$\frac{1}{2}$}
					edge from parent node[left]{slow}
				}							
				child{node[terminal]{$1$}
					edge from parent node[right]{quick}				
				}
			}
		}
		edge from parent node[right]{$0.5$}
		edge from parent node[left, yshift=0.5em, align=left]{high \\ skill}
	}
	child{node(L1)[pl1]{}
		child{node(N)[pl2,label={[xshift=2.5em, yshift=-.0em]}]{}
			child[grow=-45]{node[terminal]{0}
				edge from parent node[above right]{run}
			}
			child[grow=down]{node(L2)[pl1]{}
				child{node[terminal]{$\frac{1}{2}$}
					edge from parent node[left]{slow}
				}							
				child{node[terminal]{$-1$}
					edge from parent node[right]{quick}				
				}
				edge from parent node[above left]{ambush}
			}
			edge from parent node[left]{sneak}
		}
		edge from parent node[left]{$0.5$}
		edge from parent node[xshift=1.5em, yshift=0.5em, align=left]{low \\ skill}	}
;
\draw (H1)--(H2);
\node at ($(H1)!.25!(H2)+(1.5em,0em)$) {sneak};

\draw[dashed, bend left, color=pl1_color!50!black] (H1) to (L1);
\draw[dashed, bend left, color=pl1_color!50!black] (H2) to (L2);
}

\def\sneakingGameModification{
\tikzset{
	level 1/.style = {level distance=1.25\nodesize, sibling distance=1.5\nodesize},
	level 2/.style = {level distance = 1.25\nodesize, sibling distance=.75\nodesize},
	level 3/.style = {level distance = 1.25\nodesize},
	level 4/.style = {level distance = 1.25\nodesize}
}
\node(rootB) at ($(middle)+(\linewidth/3+0.8em,-2.5em)$) [chance,label=right:{}] {}
	child{node(H1B)[pl1]{}
		child{node(dummyB)[chance]{}
			child{node(H2B)[pl1]{}
				child{node[terminal]{}
					edge from parent node[left]{}
				}							
				child{node[terminal]{}
					edge from parent node[right]{}				
				}
			}
		}
		edge from parent node[right]{}
		edge from parent node[above left]{}
	}
	child{node(L1B)[pl1]{}
		child{node(NB)[pl2,label={[xshift=2.2em, yshift=-.5em]}]{}
			child[grow=-45]{node[terminal]{}
				edge from parent node[right]{}
			}
			child[grow=down]{node(L2B)[pl1]{}
				child{node[terminal]{}
					edge from parent node[left]{}
				}							
				child{node[terminal]{}
					edge from parent node[right]{}				
				}
				edge from parent node[above left]{}
			}
			edge from parent node[left]{}
		}
		edge from parent node[left]{}
		edge from parent node[above right]{}
	}
;

\draw[dashed, color=pl1_color!50!black] (rootB) circle (.55\nodesize);
\draw [dashed, rounded corners=10, draw=pl1_color!50!black]
	($(H1B)+(-0.6\nodesize,.6\nodesize)$) rectangle
	($(L1B)+(0.6\nodesize,-0.5\nodesize)$);
\draw [dashed, rounded corners=10, draw=pl1_color!50!black]
	($(dummyB)+(-0.6\nodesize,.6\nodesize)$) rectangle
	($(NB)+(0.6\nodesize,-0.55\nodesize)$);
\draw [dashed, rounded corners=10, draw=pl1_color!50!black]
	($(H2B)+(-0.6\nodesize,.6\nodesize)$) rectangle
	($(L2B)+(0.6\nodesize,-0.5\nodesize)$);
\node(rootC) at ($(middle)+(2em,-2.5em)$) [chance,label=right:{}]{}
	child{node(H1C)[pl1]{}
		child{node{}
		edge from parent[draw=none]{}
			child{node(H2C)[pl1]{}
				edge from parent[draw=none]{}
				child{node[terminal]{}
					edge from parent node[left]{}
				}							
				child{node[terminal]{}
					edge from parent node[right]{}				
				}
			}
		}
		edge from parent node[right]{}
		edge from parent node[above left]{}
	}
	child{node(L1C)[pl1]{}
		child{node(NC)[pl2,label={[xshift=2.5em, yshift=-.0em]}]{}
			child[grow=-45]{node[terminal]{}
				edge from parent node[right]{}
			}
			child[grow=down]{node(L2C)[pl1]{}
				child{node[terminal]{}
					edge from parent node[left]{}
				}							
				child{node[terminal]{}
					edge from parent node[right]{}				
				}
				edge from parent node[above left]{}
			}
			edge from parent node[left]{}
		}
		edge from parent node[left]{}
		edge from parent node[above right]{}
	}
;
\draw (H1C)--(H2C);

\draw[dashed, color=pl1_color!50!black] (rootC) circle (.55\nodesize);
\draw [dashed, rounded corners=10, draw=pl1_color!50!black]
	($(H1C)+(-0.6\nodesize,.6\nodesize)$) rectangle
	($(L1C)+(0.6\nodesize,-0.5\nodesize)$);
\draw [dashed, rounded corners=10, draw=pl1_color!50!black]
	($(NC)+(-0.6\nodesize,.5\nodesize)$) rectangle
	($(NC)+(0.6\nodesize,-0.55\nodesize)$);
\draw [dashed, rounded corners=10, draw=pl1_color!50!black]
	($(H2C)+(-0.6\nodesize,.6\nodesize)$) rectangle
	($(L2C)+(0.6\nodesize,-0.5\nodesize)$);
\draw [->,
line join=round,
decorate, decoration={
    snake,
    segment length=4,
    amplitude=5.0,post=lineto,
    post length=2pt }]
	($(NC)+(.75\nodesize,0.3\nodesize)$) -- ($(dummyB)+(-.75\nodesize,0.3\nodesize)$);
}

\def\noFinestPartitionExample{
\tikzset{
	level 1/.style = {level distance=0\nodesize, sibling distance=2.1\nodesize},
	level 2/.style = {level distance=1.25\nodesize},
	level 3/.style = {level distance=1.5\nodesize},
	level 4/.style = {level distance = 1.75\nodesize},
	level 5/.style = {level distance = 0.9\nodesize}
}
\node(anchor){}
child{node(A)[pl1]{}
	edge from parent[draw=none]
	child{node[pl2]{}
		child{node[pl2]{}
			child{node[missing]{}
				child{node[pl1]{}
					edge from parent[draw=none]
				}			
				edge from parent[draw=none]
			}
		}
	}
}
child{node(B)[pl1, xshift=-0.8\nodesize]{}
	edge from parent[draw=none]
	child{node[pl2]{}
		child{node[pl2]{}
			child[grow=down]{node[missing]{}
				child{node[pl1]{}
					edge from parent[draw=none]
				}			
				edge from parent[draw=none]
			}
			child[grow=south east]{node[pl2,xshift=0.125\nodesize]{}
				child[grow=down, level distance = 1.41\nodesize] {node[pl1]{}}
			}
		}
	}
}
child{node(C)[pl1, xshift=0.8\nodesize]{}
	edge from parent[draw=none]
	child{node[pl2]{}
		child{node[pl2]{}
			child[grow=down]{node[missing]{}
				child{node[pl1]{}
					edge from parent[draw=none]
				}			
				edge from parent[draw=none]
			}
			child[grow=south west]{node[pl2, xshift=-0.125\nodesize]{}
				child[grow=down, level distance = 1.41\nodesize] {node[pl1]{}}
			}
		}			
	}
}
child{node(D)[pl1]{}
	edge from parent[draw=none]
	child{node[pl2]{}
		child{node[pl2]{}
			child{node[missing]{}
				child{node[pl1]{}
					edge from parent[draw=none]
				}			
				edge from parent[draw=none]
			}
		}
	}
};

\foreach \L in {A,B,...,D}{
	\draw (\L-1-1) -- (\L-1-1-1-1);
	\draw[dotted] (\L) -- ($(\L)+(0,1.2\nodesize)$);
	\draw[dotted] (\L-1-1-1-1) -- ($(\L-1-1-1-1)-(0,0.9\nodesize)$);
}
\draw[dotted] (B-1-1-2-1) -- ($(B-1-1-2-1)-(0,0.9\nodesize)$);
\draw[dotted] (C-1-1-2-1) -- ($(C-1-1-2-1)-(0,0.9\nodesize)$);

\actInfosetDeprecated(A,D)
\actInfosetDeprecated(A-1-1-1-1,B-1-1-1-1)
\actInfosetDeprecated(C-1-1-1-1,D-1-1-1-1)
\actInfosetDeprecated(B-1-1-2-1,C-1-1-2-1)
}

\def\noFinestPartitionExampleA{
\infosetDeprecated[dashed,5,1](A-1,D-1)(0.65,0.5,0.6)()
\infosetDeprecated[dashed,5,1](A-1-1,D-1-1)(0.65,0.5,0.6)()
\infosetDeprecated[dashed,5,1](B-1-1-2,C-1-1-2)(0.65,0.5,0.6)()
}

\def\noFinestPartitionExampleB{
\draw [dashed,rounded corners, color=pl1_color!75!black]
	($(A-1)+(-0.6\nodesize,0.5\nodesize)$)--($(D-1)+(0.6\nodesize,0.5\nodesize)$)--($(D-1-1)+(0.6\nodesize,-0.8\nodesize)$)--($(C-1-1)+(-0.55\nodesize,-0.8\nodesize)$)--($(C-1)+(-0.55\nodesize,-0.6\nodesize)$)--($(A-1)+(-0.6\nodesize,-0.6\nodesize)$)--cycle;
\infosetDeprecated[dashed,5,1](A-1-1,C-1-1-2)(0.65,0.5,0.55)()
}

\def\ISOfail{
\tikzset{
	level 1/.style = {level distance=0.6\nodesize, sibling distance=5\nodesize},
	level 2/.style = {level distance=1\nodesize,sibling distance=2.5\nodesize},
	level 3/.style = {level distance=1.35\nodesize, sibling distance=1.25\nodesize},
	level 4/.style = {level distance=0.75\nodesize},
}

\node(root)[pl1]{}
	child foreach \i in {1,2}{node[chance]{}
		child foreach \j in {1,2}{node[pl2]{}
			child foreach \k in {1,2}{node[pl1]{}
				child{node[missing]{} 
					edge from parent[dotted]
				}
			}
		}
	}
;
\actInfosetDeprecated(root,root)
\foreach \i in {1,2}{
	\foreach \j in {1,2}{
		\foreach \k in {1,2}{
			\actInfosetDeprecated(root-\i-\j-\k,root-\i-\j-\k)
		}
	}
}
\foreach \i in {1,2}{
	\draw [dashed,rounded corners=10, color=pl1_color!75!black]
	($(root-\i)+(-0.45\nodesize,0.5\nodesize)$)--
	($(root-\i)+(0.45\nodesize,0.5\nodesize)$)--
	($(root-\i-2)+(0.6\nodesize,0.5\nodesize)$)--
	($(root-\i-2)+(0.6\nodesize,-0.6\nodesize)$)--
	($(root-\i-1)+(-0.6\nodesize,-0.6\nodesize)$)--
	($(root-\i-1)+(-0.6\nodesize,0.5\nodesize)$)--cycle;
}
}

\def\perfRecallGame[#1]{
\tikzset{
	level 1/.style = {level distance=#1\nodesize, sibling distance=5\nodesize},
	level 2/.style = {level distance=1.25\nodesize,sibling distance=2.5\nodesize},
	level 3/.style = {level distance=1.35\nodesize, sibling distance=1.25\nodesize},
	level 4/.style = {level distance=0.75\nodesize},
}
\node(R)[pl2]{}
	child foreach[count=\i] \actorI/\actorII in {chance/pl1, pl1/pl2}{node[\actorI]{}
		child foreach \j in {1,2}{node[\actorII]{}
			child foreach \k in {1,2}{
				node[terminal]{$u_{\i\j\k}$}
			}
		}
	}
;
}

\def\perfRecallInfosets{
\node[xshift=0.75\nodesize, yshift=0.25\nodesize] at (R) {X};
\foreach \i/\O/\A/\B/\C in {1/W/a/c/e,2/Z/b/d/f}{
	\node[xshift=0.0\nodesize, yshift=0.75\nodesize] at (R-\i) {X};
	\node[xshift=0.75\nodesize, yshift=0.0\nodesize] at (R-1-\i) {Y};
	\node[xshift=0.75\nodesize, yshift=0.0\nodesize] at (R-2-\i) {Z};

	\draw [decorate, decoration = {brace, amplitude=0.5\nodesize, mirror, raise=0.2\nodesize}]
		($(R-\i-1-1)+(-0.35\nodesize,-0.05\nodesize)$) --
		($(R-\i-2-2)+(.35\nodesize,-0.05\nodesize)$);
	\node at ($(R-\i-1-1)!.5!(R-\i-2-2)+(0,-1.1\nodesize)$) {\O};
	
	\node[xshift=0.0\nodesize, yshift=0.35\nodesize] 
		at ($(R-2)!.6!(R-2-\i)$) {\A};
	\node[xshift=0.0\nodesize, yshift=0.5\nodesize] 
		at ($(R-1-1)!.9!(R-1-1-\i)$) {\B};
	\node[xshift=0.0\nodesize, yshift=0.5\nodesize] 
		at ($(R-1-2)!.9!(R-1-2-\i)$) {\B};		
}
}

\def\perfRecallInfosetsA{
\draw [dashed,rounded corners=8, color=pl1_color!75!black]
	($(R)+(-0.45\nodesize,0.5\nodesize)$)--
	($(R)+(0.45\nodesize,0.5\nodesize)$)--
	($(R-2)+(0.6\nodesize,0.5\nodesize)$)--
	($(R-2)+(0.6\nodesize,-0.5\nodesize)$)--
	($(R-1)+(-0.6\nodesize,-0.5\nodesize)$)--
	($(R-1)+(-0.6\nodesize,0.4\nodesize)$)--cycle;
\infosetDeprecated[dashed,8,1](R-1-1,R-1-2)(.6,.6,.5)()
\infosetDeprecated[dashed,6,1](R-1-1-1,R-1-2-2)(.55,.5,.5)()
\draw [dashed,rounded corners=8, color=pl1_color!75!black]
	($(R-2-1)+(-0.45\nodesize,0.5\nodesize)$)--
	($(R-2-2)+(0.45\nodesize,0.5\nodesize)$)--
	($(R-2-2-2)+(0.6\nodesize,0.5\nodesize)$)--
	($(R-2-2-2)+(0.6\nodesize,-0.5\nodesize)$)--
	($(R-2-1-1)+(-0.55\nodesize,-0.5\nodesize)$)--
	($(R-2-1-1)+(-0.55\nodesize,0.4\nodesize)$)--cycle;
\node[xshift=3.25\nodesize, yshift=.75\nodesize] at ($(R-1)!.5!(R-2)$) {X};
\node[xshift=2.35\nodesize, yshift=.25\nodesize] at ($(R-2-1)!.5!(R-2-2)$) {Z};
\node[xshift=0\nodesize, yshift=.0\nodesize] at ($(R-1-1)!.5!(R-1-2)$) {Y};
\node[xshift=2.35\nodesize, yshift=.8\nodesize] at ($(R-1-1-1)!.5!(R-1-2-2)$) {W};

\node[below=1.9\nodesize,text=white] at (R-1-1) {\footnotesize{(X,X,Y,c,W)}};
}

\def\perfRecallInfosetsB{
\infosetDeprecated[dashed,10,1](R,R)(.55,.55,.6)()
\infosetDeprecated[dashed,10,1](R-1,R-2)(.55,.6,.55)(\footnotesize{{(X,X)}})
\infosetDeprecated[dashed,10,1](R-1-1,R-1-2)(.55,.6,.55)(\footnotesize{{(X,X,Y)}})
\infosetDeprecated[dashed,10,1](R-2-2,R-2-2)(.55,.55,.6)()
\infosetDeprecated[dashed,10,1](R-2-1,R-2-1)(.55,.55,.6)()
\infosetDeprecated[dashed,10,1](R-2-1-1,R-2-1-2)(.55,.55,.6)()
\infosetDeprecated[dashed,10,1](R-2-2-1,R-2-2-2)(.55,.55,.6)()
\draw [dashed, draw=pl1_color!75!black, bend left] (R-1-1-1) to (R-1-2-1);
\draw [dashed, draw=pl1_color!75!black, bend right] (R-1-1-2) to (R-1-2-2);

\node[xshift=1.25\nodesize, yshift=0.00\nodesize] at (R) {\footnotesize{(X)}};
\node[xshift=1.75\nodesize, yshift=0.00\nodesize] at (R-2-2) {\footnotesize{(X,X,b,Z)}};
\node(O)[xshift=1.\nodesize, yshift=1.25\nodesize] at (R-2-2) {\footnotesize{(X,X,a,Z)}};
\draw[dotted] ($(O)-(0.9\nodesize,0.25\nodesize)$)--(R-2-1);
\node[below=1.9\nodesize] at (R-2-2) {\footnotesize{(X,X,b,Z,Z)}};
\node[below=1.9\nodesize] at (R-2-1) {\footnotesize{(X,X,a,Z,Z)}};
\node[below=1.9\nodesize] at (R-1-2) {\footnotesize{(X,X,Y,c,W)}};
\node[below=1.9\nodesize] at (R-1-1) {\footnotesize{(X,X,Y,c,W)}};
}

\def\noFinestPartitionTrivial{
\tikzset{
	level 1/.style = {level distance=0\nodesize, sibling distance=1.5\nodesize},
	level 2/.style = {level distance=6\nodesize},
}
\node(anchor){}
child{node(A)[pl1]{}
	edge from parent[draw=none]
	child{node[pl1]{}}
}
child{node(B)[pl1]{}
	edge from parent[draw=none]
	child{node[pl1]{}}
};

\foreach \L in {A,B}{
	\draw[dotted] (\L) -- ($(\L)+(0,1.2\nodesize)$);
	\draw[dotted] (\L-1) -- ($(\L-1)-(0,0.9\nodesize)$);
}

\node(UL)[pl2] at ($(A)!0.25!(A-1)$) {};
\node(M)[pl2] at ($(A)!0.5!(A-1)$) {};
\node(BL)[pl2] at ($(A)!0.75!(A-1)$) {};
\node(UR)[pl2] at ($(B)!0.25!(B-1)$) {};
\node(BR)[pl2] at ($(B)!0.75!(B-1)$) {};

\actInfosetDeprecated(A,B)
\actInfosetDeprecated(A-1,B-1)
}

\def\noFinestPartitionGeneral{
\tikzset{
	level 1/.style = {level distance=0\nodesize, sibling distance=2.5\nodesize},
	level 2/.style = {level distance=3\nodesize},
	level 3/.style = {level distance=1.5\nodesize},
}
\node(anchor){}
	child foreach \L in {A,B}{ node(\L)[pl2]{}
		child{node[pl2]{}
			child{node[pl1]{}}
		}
		edge from parent[draw=none]{}
	}
;
\foreach \L in {A,B}{
	\draw (\L) -- ($(\L)+(0,2\nodesize)$);
	\draw[dotted] ($(\L)+(0,2\nodesize)$) -- ($(\L)+(0,2.55\nodesize)$);
	\draw[dotted] (\L-1-1) -- ($(\L-1-1)-(0,1.1\nodesize)$);
}
\actIS{A-1-1}{B-1-1}[0.7]
}

\def\noFinestPartitionInfosetsA{
\oneNodeIS{A}
\oneNodeIS{A-1}
\oneNodeIS{B}
\oneNodeIS{B-1}
\foreach \L/\O in {A/{A,B}, A-1/{C,D,E}, A-1-1/{F}, B/{A}, B-1/{B,C,D,E}, B-1-1/{F}}{
	\node[right=0.5\nodesize, child anchor=east] at (\L) {\footnotesize{\O}};
}
}

\def\noFinestPartitionInfosetsB{
\node(X)[chance] at ($(A)+(0,1.75\nodesize)$){};
\node(Y)[chance] at ($(B-1)+(0,1.25\nodesize)$){};
\draw[rounded corners=10, dashed, color=pl1_color!75!black, rotate around={-35:($(X)!.5!(B)$)}]
	($(X)!.5!(B)+(-2.1\nodesize,0.6\nodesize)$) rectangle ($(X)!.5!(B)+(2.1\nodesize,-0.6\nodesize)$);
\draw[rounded corners=10, dashed, color=pl1_color!75!black, rotate around={-35:($(A)!.5!(Y)$)}]
	($(A)!.5!(Y)+(-2.1\nodesize,0.6\nodesize)$) rectangle ($(A)!.5!(Y)+(2.1\nodesize,-0.575\nodesize)$);
\IS{A-1}{B-1}[0.55]
\foreach \L/\O in {X/A, A/B, A-1/{C,D,E}, A-1-1/{F}, B/A, Y/B, B-1/{C,D,E}, B-1-1/{F}}{
	\node[right=0.5\nodesize, child anchor=east] at (\L) {\footnotesize{\O}};
}
}

\def\noFinestPartitionArrow{
\node(anchor){};
\draw [->,
line join=round,
decorate, decoration={
    snake,
    segment length=15,
    amplitude=10.0,post=lineto,
    post length=5pt }]
	(-1.35\nodesize,4.5\nodesize) -- (1.35\nodesize,4.5\nodesize);
}

\def\pokerHistory{
\tikzset{
	level 1/.style = {level distance=1.25\nodesize, sibling distance=1.5\nodesize},
	level 2/.style = {level distance=1.05\nodesize, sibling distance=1.5\nodesize},
	level 3/.style = {level distance=1.25\nodesize, sibling distance=1.5\nodesize},		
}

\node(R0)[chance]{}
\foreach \round in {1,2,3}{
	child{node[pl1]{}
		child{node[pl2]{}
			child{node(R\round)[chance]{}}
		}
	}
}
	child{node[pl1]{}
		child{node[pl2]{}
			child{node(R4)[terminal]{end}}
		}
	}
;
}

\def\pokerA{
\foreach \i in {0,1,2,3}{
	\draw[dashed] (R\i-1)--($(R\i-1)+(-.45\nodesize,-.6\nodesize)$);
	\draw[dashed] (R\i-1)--($(R\i-1)+(.45\nodesize,-.6\nodesize)$);
	\draw[dashed] (R\i-1-1)--($(R\i-1-1)+(-.45\nodesize,-.6\nodesize)$);
	\draw[dashed] (R\i-1-1)--($(R\i-1-1)+(.45\nodesize,-.6\nodesize)$);
	\draw[<-, dotted, bend left] ($(R\i.north east)+(0.05\nodesize,0.05\nodesize)$) to ($(R\i.north)+(1.0\nodesize,0)$);
}
\node[right=1.0\nodesize, align=left] at (R0) {private cards\\ are dealt};
\node[right=1.0\nodesize, align=left] at (R1) {1st round of\\ public cards};
\node[right=1.0\nodesize, align=left] at (R2) {2nd round of\\ public cards};
\node[right=1.0\nodesize, align=left] at (R3) {last round of\\ public cards};
\node at ($(R4)+(0,-0.25\nodesize)$) {};
}

\def\pokerB{
\oneNodeIS{R0}[.5][0.5]
\foreach \i/\iplusplus in {0/1, 1/2, 2/3, 3/4}{
	\actIS{R\i-1}{R\i-1}[.6][0.5]
	\IS{R\i-1-1}{R\iplusplus}[.45][0.45]
}
}

\def\pokerC{
\foreach \i in {0,1,2,3}{
	\oppIS{R\i}{R\i-1}[.5][0.45]
	\oppActOneNodeIS{R\i-1-1}[.5][0.6]
}
\oppActOneNodeIS{R4}[.5][0.45][0.6][5][dashed]
}

\def\pokerD{
\publicState{R0}{R0-1}[.5][0.45]
\foreach \i/\iplusplus in {0/1, 1/2, 2/3}{
	\publicState{R\i-1-1}{R\iplusplus-1}[.5][0.45]
}
\publicState{R3-1-1}{R4}[.5][0.45]
\foreach \i/\description in {0/{pre-flop},1/flop,2/turn,3/river}{
	\draw [decorate, decoration = {brace, amplitude=0.5\nodesize, raise=0.7\nodesize}]
		($(R\i)+(0,0.25\nodesize)$) --
		($(R\i-1-1)+(0,-0.45\nodesize)$);
	\node[anchor=west] at ($(R\i)!.5!(R\i-1-1)+(1.25\nodesize,-0.05\nodesize)$) {\description};
}
}

\def\thickInfosets{
\tikzset{
	level 1/.style = {level distance=1.3\nodesize, sibling distance=2.5\nodesize},
	level 2/.style = {level distance=1.5\nodesize, sibling distance=1.\nodesize},
	level 3/.style = {level distance=1.5\nodesize, sibling distance=1.\nodesize}
}
\node(R)[chance]{}
	child{node[pl2]{}
		child[level distance=3\nodesize, sibling distance=2\nodesize]{node[pl1]{}}
		child[sibling distance=1\nodesize]{node[pl2]{}
			child{node[pl1]{}}
			child{node[pl1]{}}
		}
	}
	child{node[pl2]{}
		child[level distance=3\nodesize]{node[pl1]{}}
		child[level distance=3\nodesize]{node[pl1]{}}
	}
	child[, sibling distance=1.8\nodesize]{node[pl2]{}
		child[grow=down, level distance=3\nodesize]{node[pl1]{}}
		child[sibling distance=1.9\nodesize]{node[pl2]{}
			child{node[pl1]{}}
		}
	}
;
\IS{R-1}{$(R-1-2)!.4!(R-3-2)$}
\IS{R-3}{R-3-2}
\oneNodeIS{R}
\actIS{R-1-1}{R-2-2}[0.65][0.5]
\actIS{R-3-1}{R-3-2-1}[0.65][0.5]
}

\def\unfairMP{
\node(R)[chance]{}
	child foreach \i in {1,2}{
		node[pl\i]{}
		child foreach \j in {1,2}{
			node{}
			child foreach \k in {1,2}{
				node[terminal]{}
			}
		}
	}
;
\foreach \i/\player in {1/2, 2/1}{
	\foreach \j in {1,2}{
		\node[pl\player] at (R-\i-\j) {};
	}
}
\foreach \coords/\payoff in {1-1-1/1, 1-1-2/-1, 1-2-1/-1, 1-2-2/1, 2-1-1/-1, 2-1-2/1, 2-2-1/1, 2-2-2/-1}{
	\node at (R-\coords) {\payoff};
}
}
\tikzset{
	level 1/.style = {level distance=1.5\nodesize, sibling distance=5\nodesize},
	level 2/.style = {level distance=2.25\nodesize, sibling distance=2\nodesize},
	level 3/.style = {level distance=1.5\nodesize, sibling distance=1.\nodesize}
}
\def\unfairMPtext{
\draw[rounded corners=15, solid, color=pl1_color!75!black, rotate around={-30:($(R-1)!.5!(R-2-1)$)}]
	($(R-1)!.5!(R-2-1)+(-3\nodesize,0.6\nodesize)$) rectangle
	($(R-1)!.5!(R-2-1)+(3\nodesize,-0.6\nodesize)$);
\draw[rounded corners=15, solid, color=pl2_color!75!black, rotate around={30:($(R-2)!.5!(R-1-2)$)}]
	($(R-2)!.5!(R-1-2)+(-3\nodesize,0.6\nodesize)$) rectangle
	($(R-2)!.5!(R-1-2)+(3\nodesize,-0.6\nodesize)$);

\node[xshift=-0.5\nodesize, rotate=+29.5] at ($(R)!.3!(R-1)$) {pick on pl.1};
\node[xshift=+0.5\nodesize, rotate=-30.5] at ($(R)!.3!(R-2)$) {pick on pl.2};
\node[xshift=-0.35\nodesize] at ($(R-1)!.3!(R-1-1)$) {H};
\node[xshift=+0.35\nodesize] at ($(R-1)!.3!(R-1-2)$) {T};
\node[xshift=-0.35\nodesize] at ($(R-2)!.3!(R-2-1)$) {T};
\node[xshift=+0.35\nodesize] at ($(R-2)!.3!(R-2-2)$) {H};
\foreach \i in {1,2}{
	\foreach \j in {1,2}{
		\node[xshift=-0.35\nodesize] at ($(R-\i-\j)!.5!(R-\i-\j-1)$) {H};
		\node[xshift=+0.35\nodesize] at ($(R-\i-\j)!.5!(R-\i-\j-2)$) {T};
	}
}
}

\def\paddingBefore{
\begin{forest}
[,chance, name=root, s sep+=0.8\nodesize,
	[,pl1, name=A,
		[,chance, name=aa,
			[, pl1, name=a]		
		]		
	]
	[,pl1, name=B,
		[,chance, name=bb,
			[, pl1, name=b,]		
		]		
	]
	[$\dots$, yshift=-0.5\nodesize, edge={dashed},
		[, edge={draw=none},
			[$\dots$, yshift=0.25\nodesize, edge={draw=none}]
		]
	]
	[,pl1, name=Z,
		[,chance, name=zz,
			[, pl1, name=z]		
		]		
	]
]
\draw[dotted] (a) -- ($(a)-(0,0.9\nodesize)$);
\draw[dotted] (b) -- ($(b)-(0,0.9\nodesize)$);
\draw[dotted] (z) -- ($(z)-(0,0.9\nodesize)$);
\foreach \L/\O in
	{root/{A},
		A/{B}, B/{B}, Z/{B},
		aa/{0}, bb/{0,1},
		a/{1,...,N}, b/{2,...,N}, z/{N}
	}{
	\node[right=0.5\nodesize, child anchor=east] at (\L) {\footnotesize{\O}};
}
\node[xshift=-1.2\nodesize, yshift=0.25\nodesize] at (zz) {\footnotesize{0,1,...,N-1}};
\node [infoset, augmented, fit=(root)] {};
\node [pl1_cl_infoset, fit=(\corners{A})(\corners{Z})] {};
\node [infoset, augmented, fit=(aa)] {};
\node [infoset, augmented, fit=(bb)] {};
\node [infoset, augmented, fit=(zz)] {};
\node [pl1_cl_infoset, fit=(\corners{a})(\corners{z})] {};
\end{forest}
}

\def\paddingAfter{
\begin{forest}
[,chance, name=root, s sep+=0.8\nodesize,
	[,pl1, name=A,
		[,chance, name=aa,
			[,chance
				[$\dots$, edge={dotted}
					[,chance, edge={dotted}
						[, pl1, name=a]					
					]				
				]			
			]
		]		
	]
	[,pl1, name=B,
		[,chance, name=bb,
			[,chance
				[$\dots$, edge={dotted}
					[,chance, edge={dotted}
						[, pl1, name=b]					
					]				
				]			
			]
		]		
	]
	[$\dots$, yshift=-0.5\nodesize, edge={dashed},
		[, edge={draw=none}]
	]
	[,pl1, name=Z,
		[,chance, name=zz,
			[,chance
				[$\dots$, edge={dotted}
					[,chance, edge={dotted}
						[, pl1, name=z]					
					]				
				]			
			]
		]		
	]
]
\draw[dotted] (a) -- ($(a)-(0,0.9\nodesize)$);
\draw[dotted] (b) -- ($(b)-(0,0.9\nodesize)$);
\draw[dotted] (z) -- ($(z)-(0,0.9\nodesize)$);
\node[right=0.5\nodesize, child anchor=east] at (root) {\footnotesize{A}};
\foreach \L in {A,B,Z}{
	\foreach \ext/\O in {/B,!1/0, !11/1, !1111/{N-1}, !11111/N}{
		\node[right=0.5\nodesize, child anchor=east] at (\L\ext) {\footnotesize{\O}};
	}
}
\node [infoset, augmented, fit=(root)] {};
\node [pl1_cl_infoset, fit=(\corners{A})(\corners{Z})] {};
\node [infoset, augmented, fit=(aa)(zz)] {};
\node [infoset, augmented, fit=(aa!1)(zz!1)] {};
\node [infoset, augmented, fit=(aa!111)(zz!111)] {};
\node [pl1_cl_infoset, fit=(\corners{a})(\corners{z})] {};
\end{forest}
}

\begin{document}

    \title{Problems with the EFG formalism:\\
    a solution attempt using observations\\
    \medskip
    \large a cautionary tale of why you should use
    partially-observable stochastic games instead}
    \author{Kovarik Vojtech, Lisy Viliam\\
    vojta.kovarik@gmail.com viliam.lisy@agents.fel.cvut.cz
    }
    \maketitle
    \begin{abstract}
        We argue that the extensive-form game (EFG) model isn't powerful enough to express all important aspects of imperfect information games, such as those related to decomposition and online game solving.
We present a principled attempt to fix the formalism by considering information partitions that correspond to observations.
We show that EFGs cannot be ``fixed'' without additional knowledge about the original task, and show how to extend the EFG model under this assumption.
However, during our investigation, we ultimately concluded that a better solution is to abandon EFGs entirely and instead adopt a model based on partially observable stochastic games.
The primary contribution of the text thus lies in exposing the problems with EFGs and presenting a detailed study on introducing imperfect information by adding observations into an underlying perfect-information model.
    \end{abstract}

    \section{Introduction}\label{sec:intro}

The remainder of this paper operates under the assumption that the best way to fix the model we use to talk about imperfect information games is to extend the extensive-form game (EFG) model.
While writing this text, we came to believe that a better solution is to adopt a model based on partially-observable stochastic games.
For a definition of such a model and an overview of related literature, we invite the reader to see \cite{FOG}.

\paragraph{Why do we need a new notation?}
While the classical EFG formalism is very useful, it has grown outdated --- the recent algorithmic breakthroughs such as \cite{DeepStack} rely on decomposition, exploiting concepts such as augmented information sets and public states which aren't present in the classical EFG notation.
Coincidentally, multiagent reinforcement learning (MARL) solves similar problems as the EFG community, but there hasn't been a lot of transfer of ideas between the two areas. Making the EFG literature more accessible to the MARL community would simplify this transfer, as well as making the \emph{thinking} about the connections more efficient.
These two reasons lead us to believe it is time to revise the model we use to talk about imperfect information games.

\paragraph{The goal of this text:}
This text presents an analysis which should serve as an incremental step towards revising the model for imperfect information games.
We focus on a narrower list of desiderata, aiming for notation which defines
    \begin{itemize}
    \item augmented information sets -- being able to talk about available information even when it isn't my turn (e.g., for the construction of resolving gadget \cite{Neil_thesis}),
    \item observations -- to analyze how players receive and handle information (mainly since the approach seems natural, interesting, and relevant).
    \end{itemize}
We aim to do this in such a way that
    \begin{itemize}
    \item the new concepts fit together within one framework,
    \item the players get as much information as possible, as soon as possible (good for online play, leads to smaller subgames that are cheaper to solve),
    \item the resulting notation is familiar to people used to the classical notation.
    \end{itemize}

\subsection{Specific arguments for why the current notation doesn't suffice}

\begin{itemize}
\item There are currently definitions of augmented information sets (\cite{CFR-D}), public states (\cite{accelerated_BR}) and subgames (\cite{CFR-D}) in the literature. However, they do not necessarily always fit together (at the very least, the relationship between the definitions from \cite{CFR-D} and \cite{accelerated_BR} is not obvious).
\item These definitions currently do not behave as one would intuitively expect them to:
    \begin{itemize}
    \item One straightforward interpretation of the augmented information set definition from \cite{CFR-D} is ``$g,h\in \mc H$ belong to same $I\in \mc I^\textnormal{aug}_i$ if they have the sequence of $i$'s actions and classical information sets visited by $i$.''
    \item One way of defining public states is as ``subsets of $\mc H$ closed under the membership within the elements of $\mc I_i^\textnormal{aug}$ for each $i$''.
    \item However, this means that players will group all nodes between their action and their next information set. This partitioning leads to a strange behavior (for example, in poker, see Figure~\ref{fig:poker}).
    \end{itemize}    
\end{itemize}

\begin{figure}
\centering
\setlength{\nodesize}{2em}
\begin{tikzpicture}
\pokerHistory
\pokerA
\end{tikzpicture}
\begin{tikzpicture}
\pokerHistory
\pokerB
\end{tikzpicture}
\begin{tikzpicture}
\pokerHistory
\pokerC
\end{tikzpicture}
\begin{tikzpicture}
\pokerHistory
\pokerD
\end{tikzpicture}
\caption{\textbf{Failure of the recent EFG definitions}. From left to right:
(a): A subset of the public tree \emph{as we intuitively imagine it} in the game of limit Texas hold 'em, corresponding to a history where both players only placed the standard bet sizes, never folding or raising.
(b) and (c): Classical (solid) and augmented (dashed) information sets of player 1 and player 2.
(d): The public states \emph{actually} corresponding to these information partitions.
As we can see, there are only five public states along this history, rather than 12 as one might intuitively expect, and they do \emph{not} even correspond to the rounds of the game (which might still be acceptable).}
\label{fig:poker}
\end{figure}
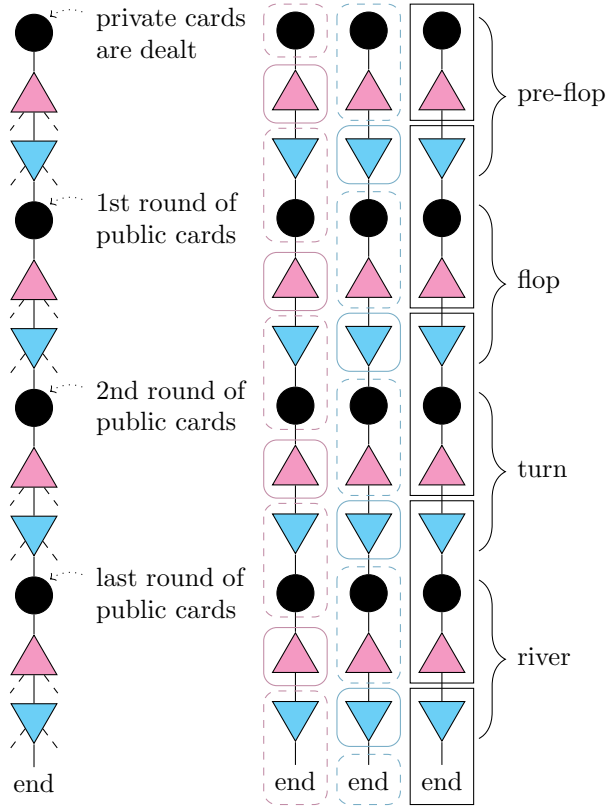

\begin{itemize}
\item We claim that the classical imperfect information EFG model is insufficient for some purposes and that the concepts such as augmented information sets either cannot be derived from it in a domain-independent way (Section~\ref{sec:ext_cl_model}).
\item While the classical model of imperfect information in games answers the question of \emph{which} information a player has available; it is under-specified in that it doesn't explain \emph{when} the information becomes available.
One way of interpreting the classical model is by assuming that the player receives no information when it isn't his turn.
However, this is often an inaccurate  way of modeling the actual situation -- when taken literally, this would correspond to each player only entering the room where the game is being played upon the start of his turn, being told the history of the game by a referee, taking his action, and then leaving the room again. Moreover, having to assume a player receives no information during his turn means he cannot use the time between his turns as effectively as he might otherwise.
\end{itemize}


In Section~\ref{sec:background}, we list the background relevant to the text that follows --- we use the perfect-information extensive-form game model as the ``underlying model'' and turn it into the ``classical'' imperfect information model by adding the ``classical'' information partitions.

Section \ref{sec:observations} formalizes the concept of observations built on top of the underlying model. We start by discussing the setting where the players aren't able to remember the observations they receive, but we quickly modify it in favor of the approach where players have perfect recall with respect to observations and their actions. We define two variants of observation histories, depending on whether the players can tell where one observation ends and another begins or not. We also introduce information partitions corresponding to observations, and the dual notion of observations corresponding to information partitions.

In Section~\ref{sec:obs_and_cl}, we present and discuss a list of technical desiderata for observations and the corresponding information partitions. We give examples of different ways in which these properties might fail to hold. We also investigate the limitations of defining observations based solely on the classical imperfect information model, without the guidance of the motivating real-life problem.

Finally, Section~\ref{sec:augm_model} introduces the proposed observation-based model of imperfect-information games. It first analyzes the notion of deducing features from observations, studies the conditions under which observation-induced information partitions have the structure of a tree, and presents a construction called ``stable modification'' that ensures these conditions are met. This section proves that the observation-based model satisfies the technical desiderata listed earlier, and gives a canonical recipe for building a ``coarse'' observation-based model that corresponds to a given classical model.

We also include an appendix, which briefly lists some ideas related to public states, common knowledge and subgames (but we do \emph{not} discuss this topic further). It also contains the description of a game of unfair matching pennies, which serves as an example of a ``tricky'' domain where some of the design proposals might be tested. We conclude with a long list of desiderata for a possible future notation.

\vknote{Not all of the proofs are currently typed in, but I have most of them written on paper. For the remaining propositions I have strong intuitions that they either hold exactly as they are written or only require minor technical modifications.}
    \section{Background}\label{sec:background}

\subsection{The Underlying Model}\label{sec:seq_model}

We can use the sequential model of EFGs, where a state in the game is represented by sequences of actions (of the players and -- if present -- chance) that lead to it being reached.
The ingredients of such $M = \left< \mc A, \mc H, N, P, \sigma_c, u \right>$ are the following:
\begin{itemize}
\item $\mc A$ is a finite set
of all \emph{actions} available in the game.
\item $\mc H \subset \left\{ a_1 a_2 \cdots a_n | \ a_i \in \mc A, \, n\in \N \cup \{0\} \right\}$ is the set of \emph{histories} legal in the game. We assume that $\mc H$ forms a non-empty finite tree\footnote{Recall that in the context of sequences, a tree is a set closed under initial segments, e.g. we have $a_1\cdots a_n \in \mc H \implies a_1 \cdots a_k \in \mc H$ whenever $k\leq n$.}.
	\begin{itemize}
	\item We use $g \sqsubset h$ to denote that $h$ \emph{extends} $g$. This is \emph{not} the strict version of extension, i.e. we have $h\sqsubset h$ for each $h$.
	\item The \emph{root} of $\mc H$ is the empty sequence $\emptyset$. The set of leaves of $\mc H$ is denoted $\mc Z$ and its elements $z$ are called \emph{terminal histories}. The set of non-leaf elements of $\mc H$ are $\mc H \setminus \mc Z$ and its elements $h$ are called \emph{non-terminal histories}.
	\item By $\mc A(h) := \{ a\in \mc A \, | \ ha\in \mc H \}$ we denote the set of actions available at $h$.
	\end{itemize}
\item The players in the game are $\{ 1, \dots, N , c\}$, where $c$ denotes the \emph{chance} (sometimes also called \emph{nature}).
\item $P : \mc H \setminus \mc Z \rightarrow \{ 1, \dots, N, c \}$ is the \emph{player function} which controls who acts in a given history.
	\begin{itemize}
	\item We use $i$ to denote elements of the player set $\{ 1, \dots, N , c\}$.
	\item Denoting $\mc H_i := \{ h \in \mc H \setminus \mc Z | \ P(h) = i \}$, we partition the histories as $\mc H = \mc H_1 \cup \dots \cup \mc H_N \cup \mc H_c \cup (\mc H \setminus \mc Z)$.
	\end{itemize}
\item $\sigma_c$ is the \emph{chance strategy} defined on $\mc H_c$. For each $h\in \mc H_c$, $\sigma_c(h)$ is a probability distribution over $\mc A(h)$.
\item $u = (u_i)_{i=1}^N$ contains the \emph{utility functions}, where $u_i : \mc Z \rightarrow \R$.
\end{itemize}

The understanding here is that the game starts in the state $h_0 = \emptyset$, then the player $P(h_0)$ picks an action $a_1\in \mc A(h_0)$ (or $a_0 \sim \sigma _c(h_0)$ is picked randomly when $P(h_0)=c$), and the game transitions to a new state $h_1 = h_0 a_1$.
This continues until the game reaches some terminal state $z = a_0 \cdots a_n$. Once this happens, each player $i \in \{1,\dots,N\}$ receives the reward $u_i(z)$.
The goal of each player is to choose actions in a way that maximizes $u_i(z)$.

\begin{remark}
Modifying EFGs formally is cumbersome.
\end{remark}
\vspace{-0.5em}
\noindent
In the underlying model where ``elements of $\mc H$ = sequences of actions'', some game-modifying operations are ill-defined or at least cumbersome to do. Examples: a) creating a subgame by replacing the upper part of the tree by a single chance node b) inserting a dummy node. The model that uses an arbitrary set $\mc H$ and transition function $\mc T : (\mc H \setminus \mc Z) \times \mc A \to \mc H$, such as POSG, would be better for this (but perhaps worse for other things).

\begin{notation}\label{not:h'}
For $\emptyset \neq h \in \mc H$ we will use the convention of rewriting $h$ as $h=h'a$, where $h'\in \mc H$, $a\in \mc A(h')$.
\end{notation}

\subsection{The Classical Definition of Imperfect Information}\label{sec:imp_inf}

The standard way of introducing imperfect information into EFGs is to define (classical) information sets. These group together all the histories that the player cannot tell apart.
Formally, a \emph{classical information partition} for player $i$ is any partition of $\mc H_i$ s.t. whenever two histories $g,h\in \mc H_i$ belong to the same element of the partition, they necessarily have $\mc A_i(g) = \mc A_i(h)$.

\begin{notation}
We use $\mc I^\textnormal{cl} = \left( \mc I^\textnormal{cl}_1, \dots, \mc I^\textnormal{cl}_N \right)$ to denote the classical information partitions (i.e. those that are always defined on $\mc H_i$ only).
For $h\in \mc H_i$ we denote by $I_i^\textnormal{cl}(h)$ the element of $\mc I_i^\textnormal{cl}$ that contains $h$.
\end{notation}

    \section{Observations and Perfect Recall}\label{sec:observations}

Suppose we have a real-life problem $P$ that includes dealing with imperfect information, and all relevant actors have perfect recall.
Optionally,  we might already have a classical imperfect information model $\left< M^\textnormal{cl}, \mc I^\textnormal{cl} \right>$ of $P$.
We use a (perfect information) EFG model $M$ to formally describe the underlying situation in $P$.
Our goal is to formalize the notions of augmented information sets and public states in $M$ in such a way that is useful for some of the recent game-theoretic algorithms and maximally advantageous for online-play.
In this section, we describe the observation formalism that shall be useful for discussing this topic.

\subsection{Observations Without Recall}

To introduce some of the ideas related to observations, we start with a brief discussion of the setting where the actors aren't able to remember their observations. Note that this setting \emph{can} be used to analyze perfect recall games since histories can be provided by the environment through observations. A secondary goal of this discussion is to allow for a comparison with the alternative proposed in Section~\ref{sec:obs_pr}.

We add observations as an extension of the underlying model. An important assumption we make is that each player knows exactly how this model works and how the generation of observations works.
At each state $h\in \mc H$, each player receives some observation. As an example, suppose player 1 observes the message ``Ugh!''. We assume the observations convey no meaning of themselves, so all that he can tell is that he is currently in one of the states which produce the observation ``Ugh!''\footnote{We could also assume that observations do convey some information about the current state of the game. This assumption would increase the ``learnability'' of $\left< M, \mc O\right>$ (see the list of desiderata in the appendix) at the cost of introducing another layer of notation.} (see the bottom left part of Figure~\ref{fig:no_memory_vs_pr}).

Formally, this can be done by using the pair $\left<M,\mc O\right>$, where $M=M^\textnormal{cl}$ and $\mc O = (O_1, \dots, O_N)$ are the observation-generating functions $O_i : \mc H \rightarrow \mathbf{O}_\textnormal{no\_recall}$ for some set $\mf O_\textnormal{no\_recall}$ of all possible \emph{observations}.
When the players have no memory, the information sets that correspond to these observations are
\begin{align*}
\{ g \in \mc H \, | \ O_i(g) = O_i(h) \}.
\end{align*}
\noindent In other words, all the information that has been observed is lost, and only the current observation is used to build the information set.

\noindent Note that the corresponding information partition covers the whole $\mc H$ -- in particular, it covers the terminal nodes $\mc Z$. All ``non-classical'' information partitions presented in this text will have this property. While the information available in terminal nodes is usually not strategically important, this aspect of the game might sometimes be interesting, for example, when the game is played repeatedly. Most importantly, it comes at no additional cost to the complexity of the discussion.

\begin{figure}[ht]
\centering 
\setlength{\nodesize}{2.5em}
\begin{tikzpicture}
\perfRecallGame[1]
\perfRecallInfosets
\end{tikzpicture}

\setlength{\nodesize}{2em}
\begin{tikzpicture}
\perfRecallGame[1]
\perfRecallInfosetsA
\end{tikzpicture}
\begin{tikzpicture}
\perfRecallGame[1.25]
\perfRecallInfosetsB
\end{tikzpicture}
\caption{Observations of the maximizing player (top, denoted X, Y, Z, W) and the corresponding memory-less information partitions (bottom left) and perfect-recall (bottom-right) information sets.}
\label{fig:no_memory_vs_pr}
\end{figure}
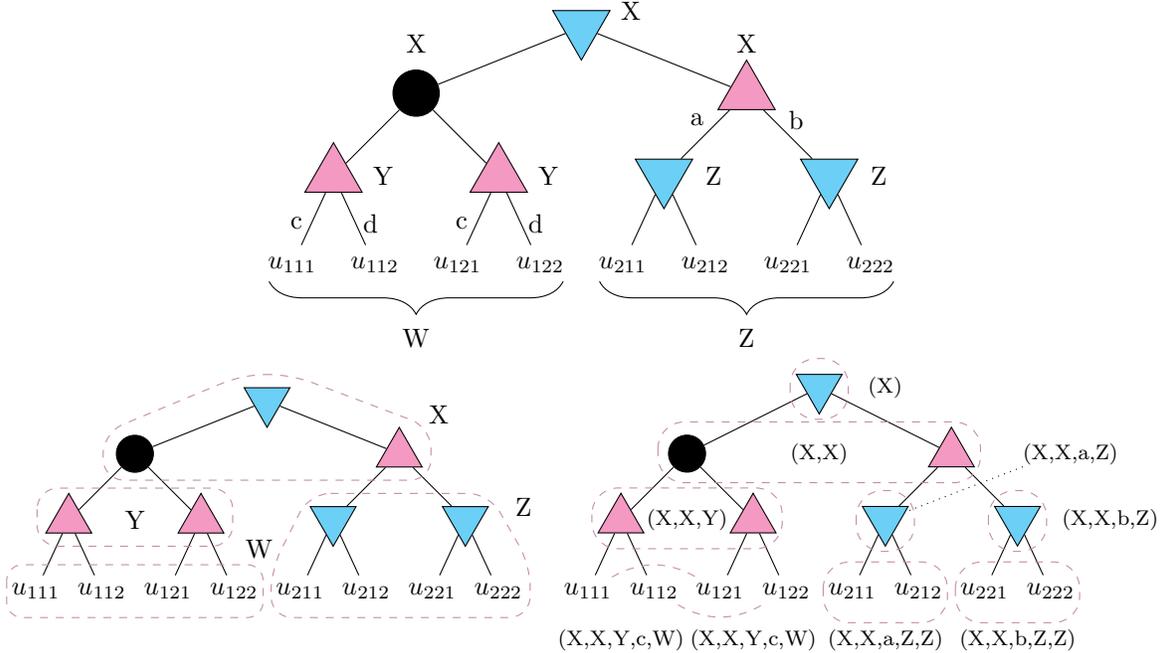

The following example shows how the classical model of imperfect information can arise as a particular case of this model.

\begin{example}\label{ex:obs_cl_no_recall}
Classical information sets coincide with those given by observing information-set labels.
\end{example}

\noindent Suppose that $\mc I_i^{\textnormal{cl}}$ is the ``classical'' information partition for player $i$ in a game $G = \left< M^\textnormal{cl}, \mc I^\textnormal{cl} \right>$.
The ``classical'' observations tell $i$ the name\footnote{The understanding here is that the observation $O=I_i^\textnormal{cl}(h)$ gives $i$ the name (label, number, hash-code) of the current information set, and he then constructs $I$ himself via $I_i(O)$, using his knowledge of $\left<M,\mc O\right>$. Formally, it might be more precise to use some injective labeling mapping $l : \mc I_i^{\textnormal{cl}} \rightarrow \N$ and write $O_i(h) = l(I_i^\textnormal{cl}(h))$. We choose not to do this for brevity.}
of the current information set whenever it is his turn to act:
\begin{equation*}
O^\textnormal{cl}_i(h) := 
\begin{cases}
\textnormal{``Not your turn.''} & h \notin \mc H_i \\
I^\textnormal{cl}_i(h) & h \in \mc H_i.
\end{cases}
\end{equation*}
Obviously, the information sets generated by $\mc O^\textnormal{cl}$ coincide with $\mc I^\textnormal{cl}_i$ in the sense that for each $h\in \mc H_i$, $\{ g \in \mc H \, | \ O_i(g) = O_i(h) \} = I^\textnormal{cl}_i(h)$.

\medskip

Note that since the agent has no memory, he relies on the observations to provide all required information. If the observations were provided verbally be a referee, he would have to repeat the whole history of the game on each turn. In practical terms, a typical element of $\mf O_\textnormal{no\_recall}$ in Chess would look like $O$ = ``Game started. White Pawn from B2 to B3. Black Pawn from B7 to B6. \dots White queen from C4 to E4. Captured black tower at E4. Your turn.''.
While this is a reasonable assumption for some practical implementations, it does not give a deeper insight into how different augmented information sets arise in games. Since the understanding of augmented information sets is the goal of this text, we will now turn our attention to the more detailed model where the agent can keep track of the observation history for himself.

\subsection{Observations with Perfect Recall}\label{sec:obs_pr}

When the players have perfect recall, the situation is different.
They can ``write down'' all of the observations they have received from the external environment so far together with all the decisions they have made, and potentially all kinds of other information (such as time-stamps for each observation, their thoughts/internal state upon making each decision, \dots).
To keep the analysis as simple as possible, we shall only consider the situation where each agent $i$ takes the observations $(O_1, O_2, \dots, O_n )$ received so far together with the actions $(a_1, \dots, a_k)$ he has made and combines them into his observation history $\vec O_i$ in a yet-unspecified manner.
The agent is then able to use $\vec O$, rather than $O_n$, to figure out where in the game tree he might currently be.

(Alternative names for $\vec O$ might be $i$'s action-observation history, information history, memory, log. We will stick with ``observation history'' in this text, with the critical disclaimer that $\vec O$ contains not only observations but also memories of past decisions, and these are of a different type than observations.)

It is surely possible to view observations as monolithic objects which cannot be further reduced (such as $O$ = ``White moved a Pawn from B2 to B3 and it is now your turn.'' without the ability to split the sentence in two). However, it will be useful to assume that there are two types of objects, elementary observations (such as $o_1$ = ``White moved a Pawn from B2 to B3.'' and $o_2$ = ``It is now your turn.'') and ``big'' observations received in a single state (such as $O = (o_1, o_2)$).

Consider an alternative scenario where the above elementary observations $o_1$ and $o_2$ were instead made in two successive states (e.g., because the opponent pressed the button on the tournament clock 30 seconds after making his move) and suppose there is a strategic advantage to being able to tell the two scenarios apart.
We need to decide whether we want our model to be such that the player can distinguish between the two scenarios \emph{based on the observations received}\footnote{Of course, we could add a third observation ``the opponent not doing anything for 30 seconds'', but that isn't the point.}?

To analyze the situation better, we shall formalize both versions and compare them. We can go ahead and spoil the conclusion: each version has its advantages and disadvantages, and they are, in general, not equivalent. However, the domains where the two approaches differ have important undesirable properties, and should be avoided (or rather, slightly different $M$ needs to be chosen to serve as an underlying model for given ``problematic'' real-life problem $P$).

\begin{definition}[Observation space]\label{def:obs_space}
The \emph{space of possible observations} in $M$ can be an arbitrary non-empty set, and shall be denoted $\mf O$.
\end{definition}

\noindent In the Chess example, typical elements of $\mc O$ can be ``Game started.'', ``White Pawn from B2 to B3.'', ``Black tower at E4 has been captured.'', ``Check by white.'', ``Game ended.'', ``You won.''. Of course, there might be different (e.g., less verbose) ways of encoding these observations. Moreover, observations can transfer all information the players might be interested in, but they can also be sparser: Suppose black knows that the white queen is at C4, a black tower is at E4, and that it is white's turn. Then the message ``C4 to E4'' suffices for him to deduce that the queen took his tower, and he doesn't have to be told explicitly.

The following ``set variant'' of observations captures the ``you can tell the difference'' approach.

\begin{definition}[Observations, set variant]
An \emph{observation function} is a vector $\mc O_{\{\}} = (\mc O_{\{1\}}, \dots, \mc O_{\{N\}})$, where $\mc O_{\{i\}} : h\in \mc H \mapsto O_{\{i\}}(h) \subset \mf O$.
The set $O_{\{i\}}(h)$ is called \emph{$i$'s observation at $h$}.
\end{definition}

By $X^* := \{ (x_1, \dots, x_n)\,|\ n\in \N\cup\{0\} ,\, x_i\in X\}$ we denote the set of all finite sequences in $X$ (including the empty sequence $\emptyset$).

\begin{definition}[Observation histories, set variant]
We define the player $i$'s \emph{observation history} function $\vec{\mc O}_{\{i\}} : h\in \mc H \mapsto \vec O_{\{i\}}(h) \in (\mc P(\mf O) \cup \mc A)^*$ as follows.
We set $\vec{O}_{\{i\}}(\emptyset) := (\, O_{\{i\}}(\emptyset) \, )$ for $h=\emptyset$ ($\vec{O}_{\{i\}}(\emptyset) := \emptyset $ when $O_{\{i\}}(\emptyset) = \emptyset$).
Otherwise, we have $h=h'a$, $a\in \mc A(h')$ and $\vec{O}_{\{i\}}(h') = ( X_1, \dots, X_n )$ for some $X_j \in \mc P(\mf O) \cup \mc A$. For non-empty $O_{\{i\}}(h)$ we set
\begin{equation*}
\vec O_{\{i\}}(h) := 
\begin{cases}
\left(\, X_1, \, \dots, \, X_n , \, a, \, O_{\{i\}}(h) \,\right) & h'\in \mc H_i \\
\left(\, X_1, \, \dots, \, X_n , \, O_{\{i\}}(h) \,\right) & h'\notin \mc H_i .
\end{cases}
\end{equation*}
Empty observations are ignored, i.e. for $O_{\{i\}}(h) = \emptyset$ we set
\begin{equation*}
\vec O_{\{i\}}(h) := 
\begin{cases}
\left(\, X_1, \, \dots, \, X_n , \, a \, \right) & h'\in \mc H_i \\
\left(\, X_1, \, \dots, \, X_n \, \right) = \vec O_{\{i\}}(h') & h'\notin \mc H_i 
\end{cases} .
\end{equation*}
\end{definition}

To formalize the approach where it doesn't have to be possible to tell where one observation ends and another begins, we define the ``sequence variant'' of the above notation.

\begin{definition}[Observation, sequence variant]
An \emph{observation function} is a vector $\mc O_{()} = (\mc O_{(1)}, \dots, \mc O_{(N)})$, where $\mc O_{(i)} : h\in \mc H \mapsto O_{(i)}(h) \in \mf O^*$.
The set $O_{(i)}(h)$ is called \emph{$i$'s observation at $h$}.
\end{definition}

\begin{definition}[Observation histories, sequence variant]
We define the player $i$'s \emph{observation history} function $\vec{\mc O}_{(i)} : h\in \mc H \mapsto \vec O_{(i)}(h) \in (\mf O \cup \mc A)^*$ as follows.
For $h=\emptyset$ we set $\vec{O}_{(i)}(h) := O_{(i)}(h) $.
Otherwise, we have $h=h'a$, $a\in \mc A(h')$ and $\vec{O}_{(i)}(h') = ( x_1, \dots, x_n )$ for some $x_j \in \mf O \cup \mc A$. For $O_{(i)}(h) = (o_1,\dots,o_k)$ we set
\begin{equation*}
\vec O_{(i)}(h) := 
\begin{cases}
\left(\, x_1, \, \dots, \, x_n , \, a, \, o_1, \, \dots, \, o_k \,\right) & h'\in \mc H_i \\
\left(\, x_1, \, \dots, \, x_n , \, o_1, \, \dots, \, o_k \,\right) & h'\in \mc H_i  .
\end{cases}
\end{equation*}
\end{definition}

\noindent Note that in the sequence variant, empty observations are ignored naturally.

When there is no danger of confusion, we omit the brackets and simply write $\mc O = (\mc O_1, \dots, \mc O_N)$, $O_i(h)$ and $\vec O_i(h)$. When we don't want to specify whether a given observation has the set or sequence form, we shall denote it as $O_i(h) := a,b,c$ --- with the understanding that it can be turned into either $\{a,b,c\}$ or $(a,b,c)$ by writing $O_{\{i\}}(h)$ or $O_{(i)}(h)$.

With the help of observation histories, players can partition $\mc H$ into information sets:

\begin{definition}[Information sets and partitions]
The \emph{information partitions} corresponding to an observation function $\mc O$ are defined as $\mc I^{\mc O} = (\mc I_1^{\mc O}, \dots, \mc I_N^{\mc O} )$, $\mc I^{\mc O}_i := \{ I_i^\mc{ O}(h) \, | \ h\in \mc H \}$, where 
\[
I_i^\mc{ O}(h) := I_i^\mc{ O} \left( \vec O_i(h) \right) :=
\{ g \in \mc H \, | \ \vec O_i(g) = \vec O_i(h) \} .
\]
\end{definition}

\medskip

Note that we do not need to use the perfect-recall model to reason about actors who have perfect recall. Indeed, the perfect-recall model can be turned into a memory-less model by replacing $\mc O_i$ by $\vec{\mc O_i}$ with $\mf O_{no\_memory} := (\mf O \cup \mc A)^*$, and the resulting model would be equally good. We do not claim that the memory-less model is worse --- rather, we claim that it only makes sense to consider memory-less models that arise in the just described manner.

\medskip

So far, we have shown how to generate information partitions from observations.
However, information sets are, in a sense, dual to observations, and there is a canonical way of going in the opposite direction:

\begin{example}\label{ex:obs_I}
Observations corresponding to an information partition.
\end{example}

\noindent Consider an arbitrary collection of partitions $\mc I = (\mc I_1, \dots, \mc I_N)$ of $\mc H$ and, writing either $h=\emptyset$ or $h=h'a$ for $h\in \mc H$ as in Notation~\ref{not:h'},  define
\begin{equation*}
O_i^{\mc I}(h) := 
\begin{cases}
I_i(h) & h=\emptyset \lor I_i(h) \neq I_i(h') \\
\emptyset & I_i(h) = I_i(h') .
\end{cases}
\end{equation*}

\noindent
To extend the definition to collections $\mc I_i$ which only cover some subset of $\mc H$, we could replace $h'$ in the above formula by ``the longest prefix of $h$ that belongs to $\bigcup \mc I_i$'' and set $O^{\mc I}_i(h) := \emptyset$ for $h\notin \bigcup \mc I_i$.

\section{Technical Desiderata and Related Obstacles}\label{sec:obs_and_cl}

When constructing an observation-based model $\left< M, \mc O \right>$ of $P$ in the (optional) presence of a classical model $\left< M^\textnormal{cl}, \mc I^\textnormal{cl} \right>$, it is natural to consider the following properties:

\begin{itemize}
\item[\namedlabel{it:noCH}{($\neg$CH)}] No changes to $M^\textnormal{cl}$ -- the underlying model shouldn't be modified in any way, i.e. $M=M^\textnormal{cl}$.
\end{itemize}
\noindent To simplify the the notation, we will henceforth assume that \ref{it:noCH} holds unless explicitly stated otherwise.
\begin{itemize}
\item[\namedlabel{it:cons}{(CONS)}] Consistency with $\left< M^\textnormal{cl}, \mc I^\textnormal{cl} \right>$ -- the information sets generated by $\mc O$ should coincide with $\mc I^\textnormal{cl}$:
\begin{equation*}\label{eq:coinc_on_H_i}
\mc I^{\mc O}_i |_{\mc H_i} = \mc I^\textnormal{cl}_i
\textnormal{ -- in other words, we have }
\left( \forall h\in \mc H_i \right) :
I_i^\mc{ O}(h) \cap \mc H_i = I_i^\textnormal{cl}(h) .
\end{equation*}
\item[\namedlabel{it:sep}{(APS)}] Acting-player separation -- each player should be able to tell whether it's currently his turn to act or not:
\begin{equation*}
\left(\forall h \in \mc H_i \right) \left(\forall g \in \mc H_{-i} \right) : I_i^\mc{ O}( h ) \neq I_i^\mc{ O}( g ) . \label{eq:separate}
\end{equation*}
\item[\namedlabel{it:iso}{(ISO)}] Immediate self-observation -- additionally to \emph{remembering} their own actions, the players should also be able to \emph{observe} their own actions as soon as they are made.
\item[\namedlabel{it:io}{(IO)}] Immediate observation -- more generally, if an event is visible to an actor in $P$, then the corresponding player in $\left<M,\mc O\right>$ should observe the corresponding piece of information as soon as the event happens.
\item[\namedlabel{it:tip}{(TSIP)}] Tree-structured information partitions -- each $\mc I_i^{\mc O}$ should form a tree when endowed with the partial order $I \sqsubset J \iff (\exists g \in I) (\exists h\in J) : g \sqsubset h$.\footnote{Without \ref{it:tip}, public states cannot serve as roots of subgames. Showing this requires, e.g., the game from Figure~\ref{fig:sneaking_game_modification}, and definition of public states and subgames (which aren't presented in this text at all).)}
\end{itemize}

\noindent As we will see in Lemma~\ref{lem:stab}, \ref{it:tip} is equivalent to the following stability condition.
\begin{itemize}
\item[\namedlabel{it:stab}{(STAB)}] Stability -- the observations generated by $\mc I^{\mc{ O}}$ should be equivalent to $\mc O$:
\begin{equation*}\label{eq:stable}
\mc I^{\mc{ O}} = \mc I^{\mc{ O^\star}}, \textnormal{ where }
\mc O^\star := \mc O^{\mc I^{\mc{ O}}} .
\end{equation*}
\end{itemize}

We make the following observations and remarks:

\begin{itemize}
\item Most of the above conditions speak about the properties of $\mc I^{\mc O}$. In this sense, the following text speaks about how to design well-behaved augmented information sets, and observations are only used as a useful formal language for the discussion.

\item This notation allows for a simple-yet-formal definition of perfect recall in the classical model:
\begin{definition}[Perfect recall]\label{def:perf_rec}
Players in $\left< M^\textnormal{cl}, \mc I^\textnormal{cl} \right>$ have \emph{perfect recall} if $\mc O^{\mc I^\textnormal{cl}}$ satisfies \ref{it:cons}.
\end{definition}

\item We do not claim that it is necessary to satisfy the condition \ref{it:iso}. Rather, our understanding is that the model shouldn't break when it is included. For example, we might wish to include the action in a public observation that is visible to everyone, and hence also to the player who has made the action.

\item \ref{it:stab} can be viewed as a stronger version of consistency of $\mc O^{\mc I^{\mc O}}$ with $\mc I^{\mc O}$ (in the sense of \ref{it:cons}).
\end{itemize}

\subsection{Observations and the Classical Model}

\begin{example}[Classical observations]\label{ex:cl}
Observations corresponding to $\mc I^\textnormal{cl}$ satisfy \ref{it:cons} and \ref{it:sep}, but fail \ref{it:io} and \ref{it:iso} (and \ref{it:stab}).
\end{example}

\noindent 
The classical information partitions are ``thin'' in the sense that it is impossible for $h=h'a$ and $h'$ to belong to the same $I\in \mc I^\textnormal{cl}_i$. It follows that $\mc O^{\mc I^\textnormal{cl}}$ can also be written as
\begin{equation*}
O_i^\textnormal{cl}(h) := 
\begin{cases}
\emptyset & h \notin \mc H_i \\
I^\textnormal{cl}_i(h) & h \in \mc H_i .
\end{cases}
\end{equation*}

\noindent These observations are consistent with $\mc I^\textnormal{cl}$ (tautologically) and satisfy \ref{it:sep} (since $h$ belongs to $\mc H_i$ iff $\vec O_i(h)$ ends with an information set).

However, they fail \ref{it:iso} --- indeed, for example in a two-player game without chance where players take one turn each, if player 1 didn't remember his action $a_1$ at some $h\in \mc H_1$, he would only learn about it through observing the next information set $I_1(h a_1 a_2)$ where he acts (and not already during the opponent's turn at $h a_1$).
More generally, $\mc O^{\mc I^\textnormal{cl}}$ fails \ref{it:io} as soon as a visible action of another player leads to some $g \notin \mc H_i$ (e.g. in any two-player perfect-information game with chance).
Later (Example~\ref{ex:iso}) we will see that $\mc O^{\mc I^\textnormal{cl}}$ also sometimes fails \ref{it:stab}.

\medskip

To simplify the notation, we introduce the notion of adding two observation functions together:
\begin{notation}
Let $\mc O$ and $\mc O'$ be two observation functions in $M$. By $\mc O + \mc O'$ we denote the observation function defined, for $O_i(h) = o_1,\dots, o_n$ and $O'_i(h) = o'_1,\dots, o'_k$, as $O^+_i(h) := o_1, \dots, o_n, o'_1, \dots, o'_k$.
\end{notation}


The observations from Example~\ref{ex:cl} can be modified to get \ref{it:iso}. However, this comes at a cost:

\begin{example}[Immediate self-observation]\label{ex:iso}
Enforcing \ref{it:iso} together with \ref{it:noCH} and \ref{it:sep} can destroy \ref{it:cons} when the set approach is adopted.
Enforcing \ref{it:noCH} together with \ref{it:cons} and \ref{it:sep} can destroy \ref{it:stab}.
\end{example}

\noindent Consider the observations $\mc O^\textnormal{iso}$ which let each player observe his action $a$ in the history immediately following the one where he took $a$:
\begin{equation*}
O_i^\textnormal{iso}(h) := 
\begin{cases}
\emptyset & h = \emptyset \lor h' \notin \mc H_i \\
a & h' \in \mc H_i .
\end{cases}
\end{equation*}
Combining these with the classical observations into $\mc O^\textnormal{iso+cl} := \mc O^\textnormal{iso} + \mc O^\textnormal{cl}$ yields observations which satisfy \ref{it:iso} (by definition) and \ref{it:sep} (by Example~\ref{ex:cl}).

However, as the game in Figure~\ref{fig:sneaking_game_modification} illustrates, $\mc O^\textnormal{iso+cl}_{\{\}}$ can be inconsistent because the player can now always tell whether he acts two times in a row or not.
The sequence version $\mc O^\textnormal{iso+cl}_{()}$ will remain consistent, but it will be unstable.
In fact, the only partition, of $M^\textnormal{cl}$ from Figure~\ref{fig:sneaking_game_modification}, satisfying both \ref{it:cons} and \ref{it:sep} is the one depicted in therein. Since this partition is unstable, we get the second part of the claim.
Adding a single dummy node (Figure~\ref{fig:sneaking_game_modification}, right) leads to both $\mc O^\textnormal{iso+cl}_{\{\}}$ and $\mc O^\textnormal{iso+cl}_{()}$ satisfying \ref{it:cons}, \ref{it:sep}, \ref{it:iso} and \ref{it:stab}, at the cost of modifying the domain (i.e. failing \ref{it:noCH}).

\begin{figure}[htb]
\centering
\begin{tikzpicture}[scale=1]
\node(middle){};
\sneakingGamePartition
\sneakingGameModification
\end{tikzpicture}
\caption{Sneaking game where the skill of one player determines whether the other will detect him, and different ways of subsequent attacking are effective in each scenario (left). This is a domain where set-like observations together with immediate self-observation are inconsistent. (Credit to Tomáš Gavenčiak for realizing this might be possible.)
Even without \ref{it:iso}, the augmented information states are unstable (middle), but a stable modification can be constructed by adding one dummy node (right).}
\label{fig:sneaking_game_modification}
\end{figure}
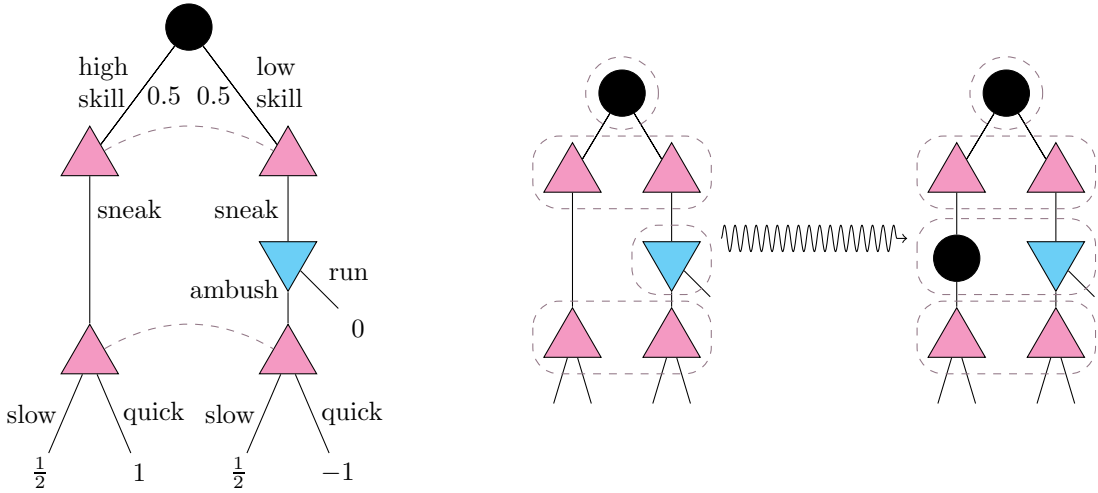

Example~\ref{ex:iso} yields the following corollary:

\begin{corollary}[A case against \ref{it:noCH}]\label{cor:impos}
In general, it is impossible to simultaneously achieve \ref{it:cons}, \ref{it:sep} and \ref{it:stab} unless \ref{it:noCH} is dropped.
\end{corollary}

In some sense, the underlying model of the sneaking game counterexample is anomalous. The following lemma characterizes those domains which do not have to be modified in order to make $\mc O^\textnormal{iso+cl}$ work:

\begin{lemma}[Well-behaved domains]\label{lem:consist_domain}
If the condition \ref{it:wb} holds for $\left< M^\textnormal{cl}, \mc I^\textnormal{cl} \right>$, then $\mc O^{\textnormal{cl+iso}}$ satisfies \ref{it:cons}, \ref{it:sep}, \ref{it:iso} and \ref{it:stab} without having to drop \ref{it:noCH}.
\begin{itemize}
\item[\namedlabel{it:wb}{(WBD)}] $\ \ \ \ \left( \forall I \in \mc I_i^\textnormal{cl} \right) \left( \forall g,h \in I \right) \left( \forall a \in \mc A(I) \right) :
\left( ha \in \mc H_i \ \& \ ga \in \mc H_i \right) \lor \left( ha \notin \mc H_i \ \& \ ga \notin \mc H_i \right)$ .
\end{itemize}
\end{lemma}

\begin{proof}
As noted in Example~\ref{ex:cl}, $\mc O^{\textnormal{cl}}$ satisfies \ref{it:cons}. If $\left< M^\textnormal{cl}, \mc I^\textnormal{cl} \right>$ has perfect recall, adding $\mc O^{\textnormal{iso}}$ adds no new information at $\mc H_i$, so consistency is preserved.
\ref{it:sep} holds since $\vec O_i(h)$ ends in an infoset if and only if $h$ belongs to $\mc H_i$ (thanks to $\mc O^{\textnormal{iso}}$, it ends in an action for $h\notin \mc H_i$).
\ref{it:iso} trivially holds thanks to $\mc O^{\textnormal{iso}}$.

To prove \ref{it:stab}, note that stability can only fail if there are two clasical infosets $I\sqsubset J$ of player $i$, such that there is no classical infoset of $i$ between $I$ and $J$, and $i$ receives a different sequence of observations between $I$ and $J$ depending on which path they take.
For $\mc O^{\textnormal{cl+iso}}$, this can only happen because of $\mc O^{\textnormal{iso}}$. Moreover, this can only happen if some path from $I$ leads immediately to $J$, while another path first visits an augmented infoset where $i$ doesn't act.
This implies the result since \ref{it:wb} ensures that such a situation cannot occur.
\end{proof}

\noindent This lemma also serves as a guide for modifying those domains where $\mc O^\textnormal{iso+cl}$ would be unstable or inconsistent.
Indeed, the obvious way of meeting the assumption \ref{it:wb} is to add a dummy chance node with a single noop action below each player node. However, we typically do not have to go so far, as it suffices to add dummy nodes after those actions in $\mc A(I)$ which sometimes lead directly to $\mc H_i$ and sometimes not.

We believe that Corollary~\ref{cor:impos} is a strong argument\footnote{However, the argument is not decisive since there might still be other ways of modeling observations that avoid this issue.} against the ``no modifying of $M^\textnormal{cl}$'' requirement. We explain that some underlying models weren't constructed with online-play in mind, and thus, one shouldn't expect every such model to be reasonable for modeling real-time decision-making. Luckily, this is not a major limitation, as the modification suggested by Lemma~\ref{lem:consist_domain} makes close to no changes to $\mc H$ in typical cases, and doubles the size of $\mc H$ in the worst case.

The last example of this section is related to the immediate observation property.

\begin{example}[Observations are domain specific]\label{ex:obs_are_dom_spec}
Whether \ref{it:io} holds depends not only on $\mc I^\textnormal{cl}$, but also on the original problem $P$.
\end{example}

\noindent We have seen that there are domains where $\mc O^{\textnormal{iso+cl}}$ satisfies \ref{it:noCH}, \ref{it:cons}, \ref{it:sep}, \ref{it:iso} and \ref{it:stab}.
However, it will often fail the more general \ref{it:io} property.
Indeed, consider a trivial betting game which starts by the referee randomly drawing and revealing one card (either A or B), then player 1 has one second to decide whether to bet 1 dollar or give up, and then player 2 has an additional one second to decide whether to bet 1 dollar or give up. If the revealed card was A, then player 1 wins all bets. If it was B, then player 2 wins all bets.

According to the description above, player two should have two seconds to think -- one during the first player's move, and one during his move. Importantly, he should know the value of the A/B card for the whole period. However, according to $\mc O^{\textnormal{iso+cl}}$, he will only have this information for one second, during his turn.
However, we could modify the game definition such that player two now has to have his eyes closed until his turn begins.
In this case, the classical model of the game would remain the same but $\mc O^{\textnormal{iso+cl}}$ would be the correct description of the situation.

The takeaway from these two examples is that when defining augmented information sets (or observations) in a specific imperfect information game, one should make design decisions based on the original problem, rather than on its classical model $\left< M^\textnormal{cl}, \mc I^\textnormal{cl} \right>$.

\subsection{Extending the Classical Model in the Absence of \texorpdfstring{$P$}{P}}
\label{sec:ext_cl_model}

Sometimes we might want to apply an algorithm such as \cite{CFR-D} or \cite{MCCR} to an imperfect information game for which we only have the classical model $G=\left< M^\textnormal{cl}, \mc I^\textnormal{cl} \right>$.
To do this, we need to define augmented information sets in $\mc H$ that are consistent with $G$.
Moreover, in resolving it is cheaper to resolve smaller subgames, so it is desirable to observe as much information as possible (and thus have smaller augmented information sets and public states).
While \ref{it:io} is ill-defined in this case (as shown in Example~\ref{ex:obs_are_dom_spec}), it might still be beneficial to have some automatic way of generating augmented information sets that are reasonably small.
An easy-to-define candidate are the ``coarse'' augmented information sets corresponding to $\mc O^{\textnormal{iso+cl}}$ (first introduced in \cite{CFR-D}\notetoself{Was it there first?}).
Unfortunately, $\mc O^{\textnormal{iso+cl}}$ fails \ref{it:io} even in the case of perfect information games (Figure~\ref{fig:coarse_infosets_fail_io}), so a different solution is required.
A natural approach would be to find the ``most informative'' observations consistent with $\mc I^\textnormal{cl}$.
However, as shown in Example~\ref{ex:no_finest_partit}, there might be several such observations functions, without a definite candidate for the ``correct'' answer.

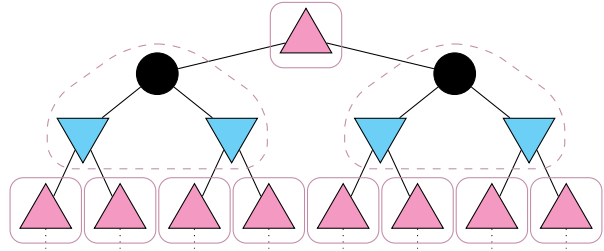
\begin{figure}[ht]
\centering 
\begin{tikzpicture}[scale=1]
\setlength{\nodesize}{2.25em}
\ISOfail
\end{tikzpicture}
\caption{A perfect information game illustrating that the information partitions corresponding to $\mc O^\textnormal{iso+cl}$ are sometimes too ``coarse'' to satisfy the immediate observation property \ref{it:io}.}
\label{fig:coarse_infosets_fail_io}
\end{figure}

\begin{example}[No canonical finest partition]\label{ex:no_finest_partit}
For some games, there is no single finest partition of $\mc H$ satisfying \ref{it:cons}.
\end{example}

\noindent The games in Figure~\ref{fig:no_finest_partition} illustrate that there doesn't always exist a unique solution to the problem of finding information partitions that are ``maximally informative'' (i.e. anything finer would fail either \ref{it:stab} or \ref{it:cons}).
On the left side of Figure~\ref{fig:no_finest_partition}, we see a trivial example of a domain with this property.
The domain on the right is a more sophisticated example that shows that sometimes a choice has to be made about which information to reveal.
It follows that it is impossible to find a general formula for ``the most informative $\mc O$''.

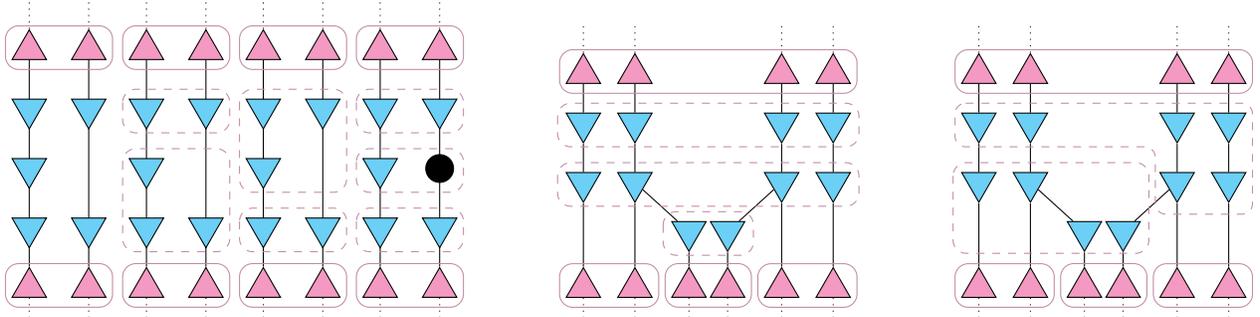
\begin{figure}
\centering 
\setlength{\nodesize}{1.5em}
\begin{tikzpicture}
\noFinestPartitionTrivial
\end{tikzpicture}
\begin{tikzpicture}
\noFinestPartitionTrivial
\augInfosetDeprecated(UL,UR)
\augInfosetDeprecated(M,BR)
\end{tikzpicture}
\begin{tikzpicture}
\noFinestPartitionTrivial
\node(dummy)[missing] at ($(B)!0.5!(B-1)$) {};
\augInfosetDeprecated(UL,dummy)
\augInfosetDeprecated(BL,BR)
\end{tikzpicture}
\begin{tikzpicture}
\noFinestPartitionTrivial
\node(dummy)[chance] at ($(B)!0.5!(B-1)$) {};
\augInfosetDeprecated(UL,UR)
\augInfosetDeprecated(M,dummy)
\augInfosetDeprecated(BL,BR)
\end{tikzpicture}
\hfill
\begin{tikzpicture}[scale=1]
\noFinestPartitionExample
\noFinestPartitionExampleA
\end{tikzpicture}
\hfill
\begin{tikzpicture}[scale=1]
\noFinestPartitionExample
\noFinestPartitionExampleB
\end{tikzpicture}
\caption{Domains where the finest partition isn't unique. Classical information sets are solid, the augmented information sets are dashed.}
\label{fig:no_finest_partition}
\end{figure}

This leaves us with three options (of which we prefer the last one).
\begin{enumerate}
\item Considering \emph{one of} the information partitions that cannot be refined further  without failing \ref{it:cons}. In general, finding such a partition might be hard.
\item Designing a heuristic that finds some reasonable refinement. This might not be worth spending time on. Note that the "let players know which information sets they might encounter next during their turn" approach is good in perfect information games and simultaneous move games, but fails already in Poker.
\item Ignore the problem, arguing that the correct approach is to consider domain-specific observations.
\end{enumerate}
    \section{The Observation-based Model of Imperfect Information}\label{sec:augm_model}

In this section, we discuss the role of observation in the absence of the classical model of imperfect information. We formalize the notion of feature-deduction features and discuss stable modifications for general observation-functions. We define the  ``official'' version of the proposed observation-based model and prove its basic properties.

\subsection{Stability and Deduction from Observations}\label{sec:features}

A \emph{feature} in $M$ is any function $f : D \to X$, where $X$ is an arbitrary set and $D =: \textnormal{dom}(f) \subset \mc H$ is the domain of $f$.
Examples of features are utility $u_i : \mc Z \to \R$, the player function $P : \mc H \setminus Z \rightarrow \{1,\dots,N,c\}$, the ``Is it my turn?'' function $\mf{1}_{\mc H_i} : \mc H \rightarrow \{0,1\}$, or the list of available actions $h\in \mc H \setminus Z \mapsto \mc A(h) \subset \mc A$.

\begin{definition}
Let $f$ be a feature in $M$ and $H\subset \mc H$.
We say that the player $i$ can \emph{deduce from observations $\mc O$ the value of $f$ at $H$} if
for each $h\in \mc H \cap \textnormal{dom}(f)$ the whole $I_i(h) \cap H$ is contained in $\textnormal{dom}(f)$ and we have $f(g)=f(h)$ for each $g \in I_i^{\mc O}(h) \cap H$.
\end{definition}

\noindent For brevity, we will say that ``$i$ can deduce $f$ at $H$'' when the context is clear, and omit the ``at $H$'' part when $H=\mc H$.

Features allows us to formulate the following characterization of stability:

\begin{lemma}[Equivalent formulations of stability]\label{lem:stab}
The pair $\left< M, \mc O \right>$ satisfies \ref{it:tip} if and only if it satisfies \ref{it:stab}.
This is further equivalent the conjunction of \ref{it:wb} and each player $i$ being able to deduce $\mc O_i$.
\end{lemma}

\noindent In particular, if $M$ satisfies \ref{it:wb}, then any  $\mc O_{\{\}}$ in $M$ satisfies \ref{it:stab}. This isn't the case for $\mc O_{()}$: While each player $i$ can, by definition of $\mc I^{\mc O}$, deduce $\vec{\mc O}_{(i)}$, he might not always be able to deduce $\mc O_{(i)}$. Apart from the above lemma, this is illustrated by, for example, the sneaking game (Figure~\ref{fig:sneaking_game_modification}).

\begin{proof}
\ref{it:tip} implies that any two histories in the same infoset have the same observation sequences, which immediately gives \ref{it:stab}. Conversely, the only way the stability condition can fail is when some two histories from the same infoset have different $\vec O_i(\cdot)$-s, which contradicts \ref{it:tip}.

Similarly, if \ref{it:tip} doesn't hold, then, for some player $i$, there must exist two different $\vec O_i(\cdot)$-s which, upon transitioning to successor histories, become identical. This can only happen if either $i$ cannot deduce $\mc O_i$ at the successor infoset, or if one transition adds an action and the other doesn't -- in other words, if \ref{it:wb} doesn't hold.
Conversely, \ref{it:wb} and being able to tell the most-recent observation imply that observation sequences that once diverged will never become the same, which yields \ref{it:tip}.
\end{proof}

As mentioned above, the obvious solution to the problem of stability is ``just use $\mc O_{\{\}}$''. However, there is an alternative solution related to Lemma~\ref{lem:consist_domain}.
By adding dummy chance nodes, each domain can be modified in such a way that it satisfies \ref{it:stab} even with $\mc O_{()}$. 
An example of how this can be done in practice is in Figure~\ref{fig:stable_modif}.
These nodes can be interpreted as a sequence of auxiliary nodes where the game-engine (or referee) processes the consequences of the latest action and optionally sends elementary observations to players.

\begin{figure}[htb]
\centering
\setlength{\nodesize}{2em}
\begin{tikzpicture}
\noFinestPartitionGeneral
\noFinestPartitionInfosetsA
\end{tikzpicture}
\begin{tikzpicture}
\noFinestPartitionArrow
\end{tikzpicture}
\begin{tikzpicture}
\noFinestPartitionGeneral
\noFinestPartitionInfosetsB
\end{tikzpicture}\caption{A domain with unstable observations (left) and the corresponding stable modification (right). Maximizer's observations are next to each history.}
\label{fig:stable_modif}
\end{figure}
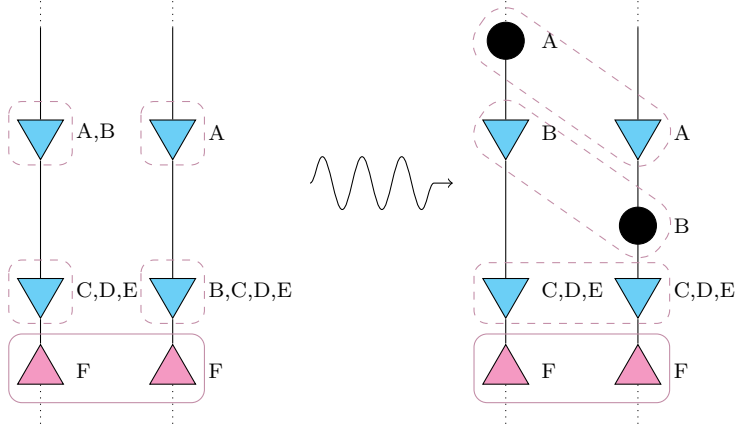

\begin{lemma}[Stable modification]\label{lem:modif_M}
For each $\left< M, \mc O \right>$ there exists a modification $\left< \widetilde M, \widetilde{\mc  O} \right>$ such that the original structure of $\mc I^{\mc O}$ is preserved, $\widetilde{\mc O}_{()}$ satisfies \ref{it:stab}, and we have $|\widetilde{\mc H}| \leq |\mc H| + |\mc H||\mc I^{\mc O}|$.
\end{lemma}

\begin{proof}[Proof sketch]\hspace{-0.25em}\footnote{Given our preference to replace EFGs by the model from \cite{FOG}, this proof is only presented as proof sketche, and the properties of the observation-based model in Section~\ref{sec:augm_mod_def} are only formulated as a conjecture.}
Note that typically, the increase in size will be much smaller than the worst-case bound. However, the bound is asymptotically tight, as illustrated in Figure~\ref{fig:stable_modif_worst_case}.
The idea is that in the most pessimistic scenario, we can split each $h$ that generates, say, an observation $O_i(h) = (a,b,c)$ into nodes $h_1$, $h_2$, $h_3$ with $O_i(h_1) = a$, $O_i(h_2)=b$, $O_i(h_3)=c$.
Since the length of observations is independent of the size of $M$, this leads to a potentially unbounded increase in size.
However, we can remove those new states $\tilde h$ where $\vec O(\tilde h)$ doesn't correspond to any original observation history $\vec O(h)$, $h\in \mc H$ (and there are at most $| \mc I(\mc O)|$ of those).
This proof-sketch can be formalized by embedding $\mc H$ into the space of sequences in $\mf O$.
\end{proof}

\begin{figure}[htb]
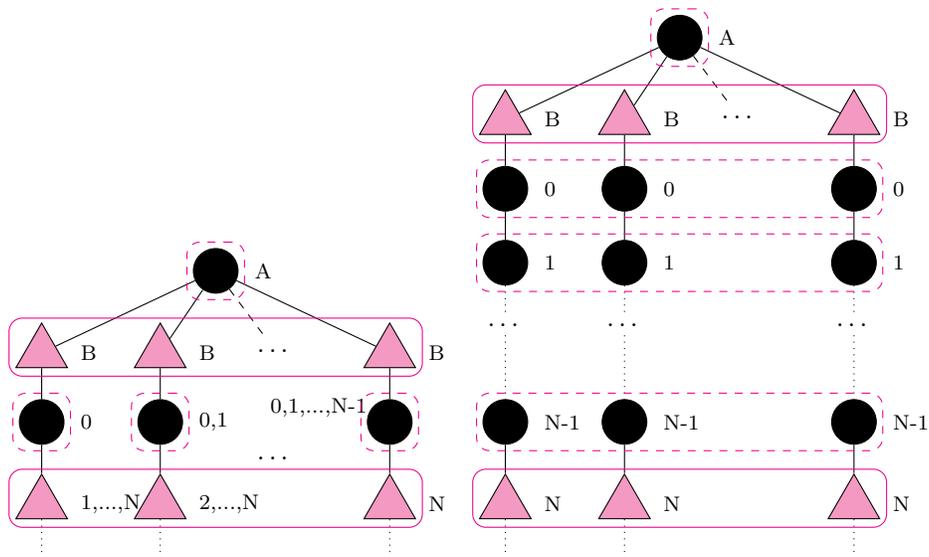

\centering
\setlength{\nodesize}{2.25em}
\paddingBefore
\paddingAfter
\caption{A domain with unstable observations (left) where the smallest stable modification (right) increases the size of $M$ from linear ($3N+1$) to quadratic ($N^2 + N +1$). A text next to a node denotes the observation received there.}
\label{fig:stable_modif_worst_case}
\end{figure}

\subsection{The Proposed Definition of the Observation-Based Model}\label{sec:augm_mod_def}

\begin{definition}[Observation-based model]\label{def:obs_model}
The pair $G = \left< M, \mc O \right>$ is said to be an \emph{observation-based model} (of an imperfect information game) if $M$ satisfies \ref{it:wb}, $\mc O$ is an observation function in $M$, and each player $i$ can deduce
\begin{itemize}
\item his own observations \hfill (the observation function ``feature'' $\mc O_i : h \mapsto O_i(h)$),
\item whether he is the acting player \hfill (the player function $P$ on $\mc H_i$ or, equivalently, the feature $\mf{1}_{\mc H_i}$),
\item his own available actions \hfill (the ``feature'' $h\in\mc H \mapsto \mc A(h) \in \mc P(\mc A)$ on $\mc H_i$),
\item whether the game has ended \hfill (the feature $\mf{1}_{\mc Z}$)
\item and his own utility \hfill (the utility function $u_i : \mc Z \to \R$).
\end{itemize}
\end{definition}

Whether an observation-based model satisfies \ref{it:io} depends on whether the choice of $\mc O$ is suitable for the specific real-life problem $P$ that $\left< M, \mc O \right>$ is supposed to model. However, all of the remaining technical desiderata from Section~\ref{sec:observations} are satisfied:

\begin{conjecture}[Observation-based model gives a well-defined game]\label{prop:obs_model_is_game}
Any observation-based model $\left< M, \mc O \right>$ has the following properties, independently of whether the set or sequence approach is adopted:
\begin{enumerate}[label=(\roman*)]
\item $\left< M, \mc O \right>$ satisfies \ref{it:sep}, \ref{it:tip} and \ref{it:stab};
\item $\mc O' := \mc O + \mc O^\textnormal{iso}$ satisfies \ref{it:iso} while being equivalent to $\mc O$ (i.e. we have $\mc I^{\mc O'} = \mc I^{\mc O}$);
\item The definition $\mc I'_i := \{ I_i \cap \mc H_i \, | \ I_i \in \mc I_i^{\mc O} \} \setminus \{\emptyset\}$ produces a well-defined classical imperfect information game $\left< M, \mc I' \right>$ with perfect recall. Moreover, $\left< M, \mc O \right>$ satisfies \ref{it:noCH} and \ref{it:cons} with respect to $\left< M, \mc I' \right>$.
\end{enumerate}
\end{conjecture}

As for how an observation-based model can be constructed in practice, consider the following examples:
adding ``Your turn!'' to $O_i(h)$ for $h\in \mc H_i$ ensures that $i$ can tell whether he is the acting player and adding the pair $(\textnormal{``Your utility''}, \, u_i(z) )$ to $O_i(z)$ for $z\in \mc Z$ ensures that $i$ can tell his own utility.
As a side effect, this automatically lets the player deduce whether the game has ended.

It is not always necessary to specify features explicitly --- since the players have access to $M$, it would suffice to for example add $(\textnormal{``Game ended at terminal node''}, \, z)$ to $O_i(z)$ for $z\in \mc Z$, and then $i$ can deduce the necessary utility himself.
However, if we want to ensure that a feature $f$ can be deduced at $H\subset \mc H$, we can always do the following:
\begin{enumerate}
\item add the pair $(f,f(h))$ to $O_i(h)$ upon first encountering a node $h$ from $H\cap \textnormal{dom}(f)$,
\item add the pair $(f,f(h))$ whenever the current node $h$ belongs to $H\cap \textnormal{dom}(f)$ and $f(h)$ differs from the value last encountered in $H \cap \textnormal{dom}(f)$, and
\item not add any ``false'' observations, i.e. if $g\in H$ satisfies $f(h)\neq f(g)$, then $(f,f(h)) \notin O_i(g)$.
\end{enumerate}

\noindent An alternative is to repeat the observation $(f,f(h))$ for every $h\in H \cap \textnormal{dom}(f)$ (cf. the Frame problem \cite{frameProblem_wiki}).


Definition~\ref{def:obs_model} lists some features that players can deduce in every imperfect information game.
Actually -- as shown by Remark~\ref{rem:non-ded_feats} -- Definition~\ref{def:obs_model} summarizes \emph{all} of the ``standard''\footnote{The authors realize that ``standard feature'' is not a well-defined notion. We apologize to any reader who finds our list incomplete.} features that the players can be allowed to deduce, because for each standard $f$ not listed in Definition~\ref{def:obs_model}, there is a counterexample domain where players cannot be consistently allowed to deduce $f$.
This is, in fact, good news, since our formalism allows us to say, for example, ``consider an observation-based model where each player can deduce X'', and thus elegantly avoid a specific class of pathological domains.

\begin{figure}[htb]
\centering
\begin{tikzpicture}
\setlength{\nodesize}{2em}
\thickInfosets
\end{tikzpicture}
\caption{An example of a game where augmented information sets of the maximizer (dashed) necessarily have to be ``thick'' in order to be both stable and consistent with the classical information sets.}
\label{fig:thick_infosets}
\end{figure}
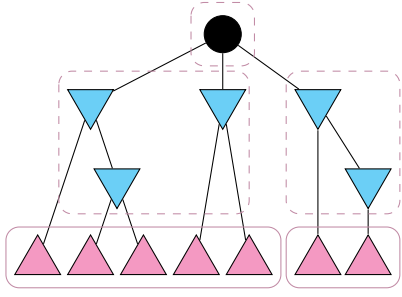

\begin{remark}[Non-deductible features]\label{rem:non-ded_feats}
For some classical imperfect information games $\left< M^\textnormal{cl}, \mc I^\textnormal{cl} \right>$ there is no observation-based model satisfying \ref{it:noCH} and \ref{it:cons} for which the players can deduce
\begin{enumerate}[label=(\roman*)]
\item the player function on the whole $\mc H \setminus \mc Z$ \textnormal{(e.g. when the opponent might, or might not, lose control of his actions and be replaced by chance)},
\item available actions on the whole $\mc H \setminus \mc Z$, \textnormal{i.e. those of the opponent or chance (e.g. most card games)},
\item history length $|h|$ \textnormal{(e.g. the game from Figure~\ref{fig:thick_infosets} or phantom variants of perfect information games, when the ``illegal move'' observations are private; in such games, information sets not belonging to the acting player might contain some $h$ together with its parent $h'$)}.
\end{enumerate}
\end{remark}

The following example shows a canonical way of extending a classical model into a consistent observation-based model.
This approach is domain-agnostic, but this comes at the cost of \ref{it:io}.
The new model can be further refined in order to get \ref{it:io}, but such a modification necessarily has to be domain specific (as noted in Example~\ref{ex:obs_are_dom_spec} and \ref{ex:no_finest_partit}).

\begin{example}[Coarse observation-based model]\label{ex:coarse_model}
For every classical model $\left< M^\textnormal{cl}, \mc I^\textnormal{cl} \right>$ of $P$, there exists $\left< M, \mc O \right>$ that satisfies \ref{it:cons}, \ref{it:sep}, \ref{it:iso} and \ref{it:stab}, but might fail \ref{it:noCH} and \ref{it:io}.
\end{example}

\noindent For a fixed $\left< M^\textnormal{cl}, \mc I^\textnormal{cl} \right>$, we first construct $M \supset M^\textnormal{cl}$ by adding dummy nodes (chance nodes with a single noop action) after some of the player nodes, in such a way that \ref{it:wb} holds.
We then define $O_i(h)$ as the observation that contains the elements enforced by the following conditions, and no others.
\begin{itemize}
\item $O_i(h) \ni (\textnormal{``Current information set: '', } I) \iff h \in \mc I^\textnormal{cl}_i, \ I = \mc I^\textnormal{cl}_i(h)$
\item $O_i(h) \ni (\textnormal{``Available actions: '', } \mc A(h)) \iff h \in \mc H_i$
\item $O_i(h) \ni (\textnormal{``Your last action: '', } a) \iff h = h'a, \ h'\in \mc H_i, \ a \in \mc A(h')$
\item $O_i(z) \ni (\textnormal{``Game ended. Utility gained: '', } u_i(z)) \iff z \in \mc Z$
\end{itemize}
With the exception of the ``$i$ can tell $\mc O_i$'' condition, it is straightforward to verify that $\left< M, \mc O \right>$ satisfies the definition of an observation-based model. The ``$i$ can tell $\mc O_i$'' condition follows from \ref{it:wb}.
It remains to show that it is consistent with $\left< M^\textnormal{cl}, \mc I^\textnormal{cl} \right>$, i.e. that $\mc I' = \mc I^\textnormal{cl}$.
This can be proven along the lines of ``isn't consistent $\implies$ fails \ref{it:wb}''.
    \section{Conclusion}

We have pointed out several issues with the standard EFG model (Figures~\ref{fig:poker} and \ref{fig:coarse_infosets_fail_io}), as well as examples where augmented information sets necessarily behave undesirably or against our intuition (Figures~\ref{fig:sneaking_game_modification}, \ref{fig:thick_infosets} and \ref{fig:unfair_mp}).
We described a model that is similar to EFG but uses observations to introduce imperfect information into an originally fully-observable environment.
Any instance of this observation-based model can be turned into a compatible EFG (Conjecture~\ref{prop:obs_model_is_game}), and the converse is also true if one is willing to modify the underlying perfect-information model and observe some information late (Example~\ref{ex:coarse_model}).
We have also shown that if one isn't willing to make these sacrifices, the task is impossible (Corollary~\ref{cor:impos} and Example~\ref{ex:no_finest_partit}).

The EFG model has not been introduced with online play, decomposition, or observations in mind.
After reviewing the difficulties that come up when attempting to extend EFGs by adding observations and staying consistent with the original model, we ultimately concluded that it is more productive to work in a model which contains observations natively.
Therefore, we believe that for investigating sequential decision-making in multiagent settings with imperfect information, researchers should adopt a model based on partially-observable stochastic games (POSGs, see, e.g., \cite{hansen2004dynamic}), such as the Factored-Observation Game model described in \cite{FOG}.
We predict this will benefit the researchers who are now working in the EFG formalism since POSGs will make decomposition easier, and simplify the transfer of ideas between game theory and multiagent RL (MARL) by using language with which the (more numerous) MARL community is more familiar. Moreover, POSG-based models offer several advantages over EFGs, such as compact representation of cyclic state-spaces, the ability to give feedback in non-terminal states, and natural modeling of simultaneous decisions.

\subsection*{Acknowledgments}

We are grateful to Tomas Gavenciak for discussions related to an earlier version of this paper.
This work was supported by Czech science foundation grant no. 18-27483Y and RCI  CZ.02.1.01/0.0/0.0/16 019/0000765.

    \bibliographystyle{alphanum}
    \bibliography{refs}

    \appendix
%
%
%
%

\section{Bonus Exercise}\label{sec:exercise}

\begin{figure}[htb]
\centering
\begin{tikzpicture}
\unfairMP
\unfairMPtext
\end{tikzpicture}
\caption{A classical model for the unfair variant of Matching Pennies.}
\label{fig:unfair_mp}
\end{figure}
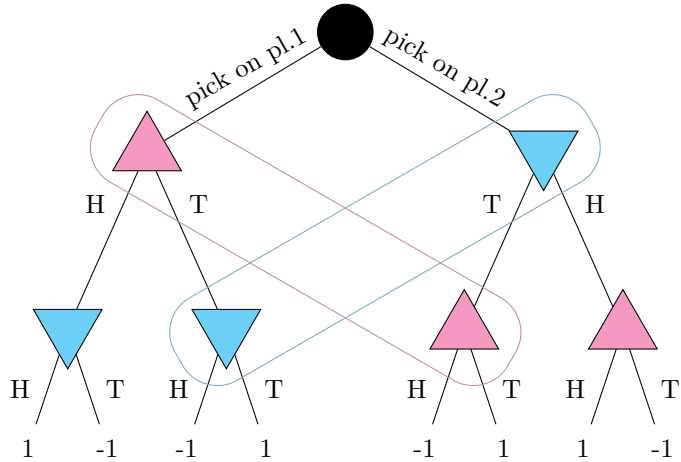

As an example illustrating the difficulties that one can encounter when defining augmented information sets and public states, we include a ``case study'' exercise.
Consider the game defined by the following description:

\begin{example}\label{ex:unfair_mp}
Unfair Matching Pennies.
\end{example}

\noindent Two players are playing a real-life game of Matching Pennies for the prize of 10\,000 dollars. To ensure that none of them can cheat, they use a third person as a referee. At the start of the game, each player is in a separate room. The referee randomly chooses the player who acts first, goes to him, and asks him for his action. He then goes to the second player and asks him for his action. By default, the players do not get to know whether they are first or second.
Now comes the twist: the referee really hates people who choose heads --- if the first player to choose picks heads, the referee will reveal this to the second player. (From some reason the players chose this referee even though they know this.)

\medskip

\begin{exercise}
What is the minimal classical model corresponding to the game from Example~\ref{ex:unfair_mp}? How do (augmented) information sets and public states look like in this $M^\textnormal{cl}$? Can players tell whether they have been betrayed by the referee or not? Is $M^\textnormal{cl}$ a good underlying model for an observation-based model of this game? If not, what is the minimal well-behaved modification $M$ of $M^\textnormal{cl}$?
\end{exercise}

\end{document}